\documentclass[12pt]{article}
\usepackage[OT1]{fontenc}
\usepackage[onehalfspacing]{setspace}
\RequirePackage{amsthm}
\RequirePackage[cmex10]{amsmath}
\usepackage[cmyk,x11names]{xcolor}
\RequirePackage[colorlinks,linkcolor=blue!70!black,citecolor=black,urlcolor=black]{hyperref}
\usepackage[top=1in, bottom=1 in, left=1.2in, right=1.2in]{geometry}
\usepackage{epsfig,amsfonts,setspace,dsfont,array,geometry,wasysym,enumerate,sgame,subcaption,comment,amssymb,graphicx,tikz,pgfplots,slantsc,ifthen,xfrac,natbib,epstopdf,eqnarray,mathtools,xspace,nicefrac}
\usepackage[shortlabels]{enumitem}

\usepackage[normalem]{ulem}
\usepackage[cmtip,all]{xy}
\usepackage[charter]{mathdesign}

\usepackage[capitalise,noabbrev]{cleveref}
\crefname{section}{section}{sections}
\newtheorem{innercustomgeneric}{\customgenericname}
\providecommand{\customgenericname}{}

\theoremstyle{definition}

\newtheorem*{example*}{Example}
\newcommand{\argmax}{\operatornamewithlimits{argmax}}

\usepackage{array}
\usepackage{scalerel}
\usepackage{float}
\usepackage{ulem}
\usepackage{amsthm}
\theoremstyle{plain}
\newtheorem{thm}{Theorem}
\newtheorem{lem}{Lemma}

\newtheorem{prop}{Proposition}
\newtheorem{cor}{Corollary}

\theoremstyle{definition}
\newtheorem{defn}{Definition}

\theoremstyle{remark}

\usepackage{graphicx}

\begin{document}

\begin{titlepage}
\vspace{-3cm}
\title{{\bf \large{Optimal Refund Mechanism with Consumer Learning}}\footnote{This paper was previously circulated under the title, ``Learning Deterrence vs. Learning Encouragement: Optimal Pricing and Return Policy.'' }
}
\vspace{3cm}
\author{
\vspace{0.5cm}
\begin{minipage}{0.4\textwidth}\centering  
Qianjun Lyu\footnote{Institute for Microeconomics, University of Bonn. I thank Wing Suen for his insightful comments and for ongoing discussions. I also thank the editor and two anonymous referees, whose suggestions significantly improved the paper. I would like to give special thanks to Pak Hung Au, Sarah Auster, Francesc Dilmé, Carl-Christian Groh,  Bard Harstad, Wei He, Jin Li, Yunan Li, Xianwen Shi, Bal\'{a}zs Szentes, Dong Wei, and Yimeng Zhang for their valuable comments and suggestions at different stages of this project. I thank seminar and conference participants at the University of Hong Kong, the University of Bonn,  LSE, the University of Nottingham, and Econometric Society  meetings for helpful comments. Finally, I acknowledge funding from the Deutsche Forschungsgemeinschaft (DFG, German Research Foundation) under Germany's Excellence Strategy (EXC 2126/1--390838866).  
 }  
\\ \centering  \small \it University of Bonn
\end{minipage}  
}

\date{\vspace{0.8cm} \today}
\maketitle
\thispagestyle{empty}
\vspace{-0.8cm}
\begin{abstract}

\noindent This paper studies the optimal refund mechanism when an uninformed buyer can privately acquire information about his valuation of a product over time. We consider a class of refund mechanisms based on \textit{stochastic return policies}: if the buyer requests a return, the seller will issue a (partial) refund while allowing the buyer to keep the product with some probability. Such return policies can affect the buyer's learning process and thereby influence the return rate. 
Nevertheless, we show that the optimal refund mechanism is deterministic and takes a simple form: either the seller  offers a sufficiently low price and disallows returns to deter buyer learning, or she offers a sufficiently high price with free returns to implement maximal buyer learning.  The form of the optimal refund mechanism is non-monotone in the buyer's prior belief regarding his valuation. 
\vspace{0.5cm}

\vspace{0.5cm}

\end{abstract}

\end{titlepage}

\section{Introduction}
In 2023, approximately 14.5\% of total retail sales in the United States resulted in returns. The return rate for online purchases was even higher, at 17.6\%.\footnote{Source: Appriss Retail and National Retail Federation.
https://cdn.nrf.com/sites/default/files/2024-01/2023\%20Consumer\%20Returns\%20in\%20the\%20Retail\%20Industry.pdf
} 
Since e-commerce retailers often provide generous return policies, 
consumers often make purchases in order to learn their own valuation of a product (i.e., to learn whether the product meets their preferences); 
for example, a 2022 survey of 9,286 US adults showed that 75\% of consumers returning products purchased online did so because of poor fit.\footnote{Source: Power Reviews, https://www.powerreviews.com/research/reduce-returns.
}
On the other hand, there are many areas of commerce in which sellers offer no-return policies or returns with partial refunds. For instance, heavily discounted items are often final-sale and cannot be returned, 
and travel agencies typically charge a fixed fee for ticket refunds.

In this paper, we study the revenue-maximizing refund mechanism by explicitly examining how the design of a return policy influences a buyer's learning process. 
Our model includes an important departure from most of the existing literature on refund contracts: we consider non-deterministic (\emph{stochastic}) return policies. Previous papers have assumed that the buyer learns his valuation of a product immediately upon purchase (\citet{courtyli2000}; \cite{matthews2007}; \cite{hinnosaar2020}). Under this assumption, it is without loss of generality to restrict one's attention to deterministic return policies, in which returns are either always allowed or always disallowed. By contrast, our work is motivated by the purchase of experience goods, for which a buyer must learn his valuation gradually over time. For instance, the comfort and performance of a mattress, an office chair, or a car can only be determined through sustained use. 
In such scenarios, non-deterministic return policies may introduce another avenue for price discrimination, as they can either extend or shorten the buyer's gradual learning process.

Consider a (female) seller selling a product to a (male) buyer who values the product at either 0 or $v>0$. Neither the seller nor the buyer knows the actual valuation, but they share a common prior belief $\mu_0$ (representing the probability that the valuation is high). The buyer can learn both before and after purchase.
We model the buyer's learning process à la \cite{KRC}.  
The buyer exerts some costly effort to learn. 
If his valuation is high, then good news arrives according to a Poisson process; 
if his valuation is low, then no news ever arrives and his belief gradually declines. This is known as the ``no news is bad news'' model, or the good-news model. 
For the sake of elaboration, we assume that if the seller does not allow returns, then the buyer learns before purchase, whereas if returns are allowed, the buyer learns only after purchase. The seller's opportunity cost is assumed to be zero.

For the seller, it would be ideal to charge the price $\mu_0 v$ and disallow returns, thus extracting all of the surplus.  However, this strategy is not feasible: the buyer can always acquire information before purchase, and his option value from learning would prevent the seller from extracting all of the surplus. Consequently, the seller's expected revenue is much smaller than $\mu_0 v$. 

One possible strategy the seller can adopt is to post a  non-refundable but attractive enough price to compensate the buyer's option value, so that he is willing to buy the product immediately and forgo the benefits of pre-purchase learning.  We refer to this mechanism  as \textit{learning deterrence}: It specifies a sufficiently low price and at the same time disallows a return to prevent the buyer from private learning. 
Thus, the buyer purchases with probability one at the prior belief $\mu_0$.

The other strategy is to allow free returns, which enable the buyer to benefit maximally from post-purchase learning (i.e., he can learn until his continuation value from learning reaches zero). 
Under this strategy, the seller can charge a very high price, since she is more likely to sell to a high-valuation buyer. The downside is that the probability of sale is lower than with the first strategy, since buyer learning may lead to more returns.  
When free returns are available, the buyer's learning process is as follows. As long as no news arrives, his belief declines.  If it becomes sufficiently pessimistic, he decides to quit learning and request a refund; such belief is called the \emph{quitting belief}, $q$.  On the other hand, if good news arrives before his belief has fallen to $q$, then he decides to keep the item (and not request a refund).

These two strategies correspond to the two extremes of buyer learning: with learning deterrence, the buyer does not learn at all; with free returns, he learns as much as possible (until his belief reaches the quitting belief $q$). The seller can also combine the two strategies to induce some intermediate amount of learning. That is, she can design a \emph{stochastic return policy} appropriately, under which the buyer will choose to stop learning earlier, at a more optimistic stopping belief $\beta>q$.  

A stochastic return policy specifies two things: 
(a) the probability that the seller allows the buyer to \textit{keep} the product when he requests a refund;
(b) the amount of the refund that will be issued. 
One can interpret a stochastic return policy as a convex combination between allowing the buyer to keep the product at a (discounted) price; and asking the buyer to return the product to get a (partial) refund. 
Such an intermediate return policy can provide sufficient incentive for the buyer to stop learning earlier; in fact, we show that the seller can induce any stopping belief between the quitting belief $q$ and the prior belief $\mu_0$ through an appropriately designed stochastic return policy.

Our first result (Proposition \ref{step1}) is a characterization of the optimal return policy when the price is exogenously given.  We find that for moderate prices (neither too high nor too low), a stochastic return policy is optimal. However, our main result (Theorem \ref{deterministic}) shows that when the seller optimizes over prices, her expected revenue is quasi-convex in price (conditional on the return policy being optimally designed for each price). This implies that the optimal price must be either very high or very low. For such a price, the optimal return policy is necessarily deterministic. 
In other words, even when we look at the larger class of stochastic refund mechanisms, the optimal refund mechanism is still deterministic: if the seller provides refunds at all, she requires the buyer to return the product with probability one. This mirrors the simplicity of real-world return policies.

Interestingly, the suboptimality of stochastic return policies is driven by their flexibility in shaping the buyer's learning behavior. The optimally designed stochastic return policy allows the seller to simultaneously increase the price and decrease the buyer's stopping belief. A decrease in the stopping belief increases the probability of a successful sale (because the buyer learns for a longer time) but decreases the seller's revenue in the event of a return (because the buyer requests a return at a more pessimistic belief). 
If we interpret the probability of a successful sale as the demand for the product, then with the optimal  design of a stochastic return policy, an increase in price leads to higher demand---an exception to the law of demand.

In addition, note that the seller's expected revenue depends on the price, the net return revenue (i.e., the price net the refund), and the probability of receiving these amounts. The price represents the seller's revenue from a high-type buyer (since only a high-type buyer will ultimately keep the product without requesting a return), while the net return revenue is the amount she can obtain from either type of buyer (since she receives it even if the buyer requests a return). 
Hence, the seller can adopt either a niche-marketing strategy (aimed solely at the high-type buyer) or  a mass-marketing strategy (aimed at both types of buyer). With the former, she tends to charge a higher price in order to gain as much as possible from the high-type buyer; simultaneously, she utilizes a stochastic return policy to increase the probability of selling to a high-type buyer, which reinforces her incentive to focus on the niche market and raise the price. In contrast, with the mass-marketing strategy, the seller aims for a higher return revenue so as to gain more from both buyer types. To achieve this, she must decrease the price, which amplifies her incentive to focus on the mass market and increase the return revenue. This in turn reinforces her incentive to decrease the price. The self-reinforcing nature of these two strategies explains the above-mentioned quasi-convexity of the seller's revenue as a function of price; in particular, it explains why the optimal prices lie at the boundaries.

Our second main result (Theorem \ref{two}) characterizes the optimal refund mechanism for every prior belief  $\mu_0$ between zero and one. The mass-marketing form of the optimal refund mechanism (learning deterrence---a low price with no returns) leads to welfare maximization, while the niche-marketing form (a high price with free returns) leads to a welfare loss. 
The form of the optimal mechanism is 
a non-monotone function of the prior belief $\mu_0$: learning deterrence is optimal for small and large values of $\mu_0$, and free returns are optimal for intermediate values. This is mainly driven by the fact that information is more valuable when the prior belief is
more uncertain. 
Recall that to deter learning, the seller must lower the price sufficiently to compensate for the buyer's option value from learning. This strategy cannot be optimal unless the buyer's option value is low, i.e., unless the prior belief is very low or very high.

\textbf{Related literature.} Our paper relates to the sequential screening literature pioneered by \citet{courtyli2000}, in which the seller offers a menu of refund contracts designed to elicit the buyer's ex-ante private information in order to facilitate price discrimination. Building on this, \citet{krahmerstrauzs}  and \citet{Bergemannexpost} impose ex-post participation constraints to study the optimal sequential mechanism. We deviate from that line of research and instead study how the design of the refund contract affects the buyer's endogenous learning. In our model, the buyer's ex-post participation can be a consequence of the optimal refund mechanism.

As noted earlier, the existing literature on product returns with consumer learning typically focuses on deterministic return policies, assuming that the buyer learns his actual valuation immediately after purchase (\citet{matthews2007}; \cite{Jonas}; \citet{Janssen}).   \citet{matthews2007}  characterize the seller's optimal choice of price and refund in a setting in which the buyer can acquire perfect information at a fixed cost before purchase. 
  \Citet{Jonas} instead considers arbitrary information structures and shows that any split of the surplus between buyer and seller can be achieved. 
\citet{Janssen} discuss equilibrium refunds in the context of a market with consumer search.  Unlike all of these papers, we consider a setting with imperfect learning, which makes it viable for stochastic return policies to affect the buyer's learning behavior. 

 \cite{hinnosaar2020} investigate refund mechanisms in an environment where the buyer's signal structure is designed to minimize the seller's profit. 
 They find that the seller can best hedge against uncertainty by offering the ``robust refund policy,'' which is a mixture of generous refunds with randomly discounted non-refundable prices. 
  In contrast, our buyer's objective is to maximize his own surplus from trade net his information cost. For this objective, only the expected price of the product matters; that is, random pricing cannot affect the buyer's learning strategy.\footnote{The robust refund mechanism of \cite{hinnosaar2020} coincides with our optimal mechanism if the restocking fee in their model is 0 and the learning cost in our model is 0. } On the other hand, a random allocation rule (the probability specified in a stochastic return policy) is crucial in our model. 

Our paper is also related to  \citet{board} and \cite*{daleygreen}, which investigate option contracts that the winning bidder can choose whether to execute an option after collecting new information. In \citet{board}, a winning bidder can choose whether to use the asset at a contingent fee or to give up the upfront payoff and quit the market. 
In \cite*{daleygreen}, which discusses due diligence in mergers and acquisitions, the acquirer agrees on a price with the target firm but then has the option not to execute the contract. Both papers focus on deterministic execution. Our model, in contrast, allows stochastic execution. 

Our paper also contributes to the literature on mechanism design with information acquisition. 
 \cite{shi2012} and \citet*{mensch} study mechanism design in settings where the buyer can privately acquire information. 
\cite{shi2012} uses rotation-ordered information structures to model buyer learning. \citet*{mensch} discusses flexible information acquisition, with cost as the expected difference in a posterior-separable measure of uncertainty. 
In contrast, we investigate sequential Poisson learning. The Markov nature of this learning process allows us to analyze how the optimal refund mechanism varies with the buyer's prior belief, a question that \citet*{mensch} cannot accommodate.\footnote{This is because, with uniformly posterior-separable information costs, the same Blackwell experiment may have different costs at different prior beliefs, as discussed in Appendix A of \citet*{mensch}. Moreover, Poisson learning involves unbounded experiment, which is often ruled out in this literature; see \citet*{costofinfo}.}  

Finally, there is a growing literature on seller pricing decisions in light of buyer learning.\footnote{For a robustness perspective, see \cite{johnson-myatt} and \citet*{balazandroesler}.
} 
In relation to sequential buyer learning,\footnote{See \citet*{bonatti2011}, \citet*{bergemannvalimaki2000},
and \citet*{bv1996}.} \cite{Branco},
\citet*{lang2019}, and \citet{pease2018} discuss the seller's optimal pricing rule assuming the buyer can sequentially acquire information before trading. These papers model the buyer's learning process as Brownian motion with different state-contingent drifts.\footnote{In \citet*{lang2019}, the drift equals the buyer's actual valuation, which is distributed on a continuous support. In \cite{pease2018}, the actual valuation is binary, and high (low) value induces upward (downward) drift. } \cite{Branco} and \citet*{lang2019} conclude that the probability of sale is increasing in the buyer's ex-ante expected valuation. Superficially, this appears to be at odds with the finding in our paper that the probability of sale  is non-monotone in the buyer's prior belief. The subtle difference is that, in our model, the buyer's prior belief captures the amount of information he initially has (i.e., his initial level of certainty about his valuation); this quantity is non-monotone with respect to the belief, since very low (high) beliefs reflect greater certainty that the valuation is low (high). In \cite{Branco} and \citet*{lang2019} this channel is absent, since the buyer's expected valuation (the mean of the normal prior distribution) is unrelated to the informativeness of the prior distribution (in this case, the variance). 
Similarly, \citet{pease2018} predicts that the optimal pricing obeys a cutoff rule with respect to the buyer's prior belief; that is, if the prior belief is above some threshold, then the seller sets  higher prices to encourage learning. By contrast, in our case, the price  under the optimal refund mechanism is non-monotone in the prior belief. We explain this in more detail in Section \ref{opm}. The key issue is that in the papers cited above, the only choice variable is the price; thus, the impact of the return policy cannot be measured.

In terms of methodology, we reframe our mechanism design problem as an information design problem. Specifically, we characterize the implementable mechanisms, then constructing a transfer function on the belief space. The seller's objective is to maximize the expected value of the transfer function over all feasible distributions of posterior beliefs.\footnote{The ``feasible'' distributions are those that correspond to a learning strategy from the ex-ante point of view. } We then utilize the concavification technique in Bayesian persuasion (\cite{KamenicaGentzkow2011}) to facilitate our analysis and characterize the optimal information structure.\footnote{A similar approach is adopted in \citet*{mensch} and \cite{hinnosaar2020}.  }

The remainder of this paper is organized as follows. In Section \ref{sec:model} we present our model and some preliminary results needed to characterize the implementable mechanisms. 
In Section \ref{Exogenousprice} we study the optimal  return policy for an exogenously fixed price. 
In Section \ref{endogenousprice} we endogenize the price and identify the optimal refund mechanism, which turns out to be deterministic. In Section \ref{opm} we characterize the optimal refund mechanism for each value of the prior belief.  
In Sections \ref{posteff} and \ref{bad} we consider two extensions of the main model. In Section \ref{sec:discussion} we provide an alternative interpretation of our model and discuss the robustness of our notion of a refund mechanism.

\section{Model}
\label{sec:model}
A seller sells one unit of an indivisible product to a risk-neutral buyer, who values the product at either 0 or $v>0$.\footnote{Our results remain the same if the buyer has a binary valuation $\{v_h,v_l\}$ where  $v_h>v_l\ge 0$.} 
Both parties have a common prior belief $\mu_0\in (0,1)$, representing the probability of the high valuation $v$. The buyer's posterior belief evolves over time and is therefore denoted by $\mu(\tau)$, where $\tau$ represents time; when appropriate, we may simply write $\mu$. 
Certain specific values of the buyer's posterior belief (namely, his stopping beliefs) will be denoted by $\beta$; the distinction between $\beta$ and $\mu$ will become clear later.
We focus on the scenario in which welfare maximization requires trade with probability one; therefore, we normalize the seller's opportunity cost to 0.\footnote{Such normalization is valid as long as the seller's opportunity cost is lower than the value $v_l$ mentioned in footnote 10. } There is no cost associated with production or return. We assume that neither party discounts over time.\footnote{This assumption is not crucial to the analysis, as the seller obtains an upfront payment even with a free-return policy. Nevertheless, we retain the assumption because the time between purchase and return is usually not very long.}
The buyer's outside option is 0.

The seller commits to a refund mechanism  consisting of the product's price, denoted by $t_b$, and a return policy, denoted by $(x_r,t_r)$. The price $t_b\ge0$ indicates the transfer made by the buyer to the seller upon purchase.\footnote{The price is constant over time, because we assume the seller does not know when the buyer observes the mechanism. } 
The return policy has two elements: (a) the probability, $x_r\in [0, 1]$, that, if the buyer requests a return, then the seller will issue a refund yet allow the buyer to keep the product; (b) the \emph{return transfer}, $t_r\in [0, t_b]$, which is the (expected) net payment from the buyer to the seller if he requests a return. 
Under a typical refund mechanism $m=\{t_b, (x_r,t_r)\}$,  the buyer pays $t_b$ at purchase. Should he later request a return, the seller employs a public randomization device: with probability $x_r$, the buyer keeps the product; with the remaining probability, he must return the product to the seller. In either case, the seller issues a refund of $t_b-t_r$.\footnote{Note that although the seller could base the refund on whether the buyer returns the product, for a risk-neutral buyer, only the expected refund matters.} 
The pair $(x_r,t_r)$ can represent  various return policies. For instance, $(1,t_b)$ represents a no-return policy, and $(0,0)$ represents a free-return policy. 
We call a return policy $(x_r,t_r)$ with $x_r\in (0,1)$ a \textit{stochastic return policy}.

In addition, if returns are allowed, the seller can commit to a \textit{return window} $T\in [0,+\infty)$, starting from the purchase time, such that the return policy $(x_r,t_r)$ is only valid during this window. Note that, a priori, the seller can also limit the buyer's learning process via the return window $T$. However, she cannot gain from this additional instrument.\footnote{Loosely speaking, suppose the buyer learns only after purchase. Then the seller can induce the same learning behavior using either the return policy $(x_r,t_r)$
or the return window $T$. However, using a stochastic return policy to induce the same stopping belief can increase the allocation surplus upon return, giving the seller a higher return transfer. By contrast, using $T$ does not create surplus.} 
Therefore, without loss of generality, we can always set the return window to match the time at which it is optimal for the buyer to stop learning, given the mechanism $m$. 
We discuss this further at the end of the paper.  Henceforth, we omit the notion of $T$.

The buyer's payoff is realized when he makes his final decision either to keep the item or to request a return.\footnote{In principle, the buyer can decide at any time within the return window. We assume his payoff is realized when he chooses to stop learning and make a decision; however, since there is no discounting, one could equivalently assume the payoff is realized at $T$. }  At belief $\mu$, if the buyer decides to keep the item (i.e., not to request a return), he obtains a consumption utility $\mu v-t_b$;\footnote{One can consider $v$ as the reduced-form net present value for the product if the consumption utility is realized over time.  }
if he requests a return, he obtains a return utility $\mu v x_r-t_r$. 
Without loss of generality, we assume $v-t_b>v x_r-t_r$,  i.e., a high-type buyer does not benefit from requesting a return.\footnote{Note that if requesting a return were preferable at belief 1, then it would be preferable at every belief in $[0,1]$; that is, the buyer would request a return regardless of the learning outcome. Hence learning would not be necessary. The seller could therefore let $x_r=1$ and increase $t_r$ to get a higher revenue. } The seller receives the price $t_b$ if the buyer decides not to request a return, and the return transfer $t_r$ if he does request a return. Her ex-ante expected revenue is contingent on the buyer's learning strategy.

The buyer can learn both before and after purchase. 
As described in the introduction, we model his learning using the exponential bandit framework, with a ``no news is bad news" structure. The buyer incurs  a  flow cost $k$ to learn. If his valuation is high, then good news arrives according to a Poisson process. If his valuation is low, then no news ever arrives. 
Let $\lambda_B$ denote the pre-purchase learning rate and $\lambda_P$ the post-purchase learning rate.
We assume $\lambda_P \geq \lambda_B$, since it is reasonable to suppose that information available before purchase remains accessible after purchase. 
In fact, for most of this paper we focus on the scenario in which $\lambda_P = \lambda_B = \lambda$. This is justified by the fact that, in recent years, the gap between pre-purchase and post-purchase information has narrowed significantly, thanks to the increasing availability of information on the internet, as well as policies at some retailers (including Apple) that enable consumers to test products in person.
In Section \ref{posteff}, we extend our analysis to scenarios in which $\lambda_P > \lambda_B$, and we explore the limiting case where $\lambda_P \rightarrow \infty$, which captures the scenario in which the buyer learns his true valuation instantly after purchase.

Throughout our analysis, we maintain the assumption $4k < \lambda v$ to avoid the trivial situation in which the learning cost $k$ is so high that the buyer never chooses to learn, regardless of the mechanism or the prior belief. As long as the buyer continues to learn and good news does not arrive (i.e., no Poisson jump occurs), his belief gradually declines according to the following law of motion:
\begin{equation}\label{lawofmotion}
 \mu'(\tau)=-\mu(\tau)(1-\mu(\tau))\lambda<0.   
\end{equation}
If good news arrives, his belief jumps to one. 

Note that the buyer can always behave as if returns are disallowed and learn as much as he likes before purchase. Therefore his option value from pre-purchase learning is endogenously determined through the price $t_b$ alone. 
We denote this option value by $V^0(\mu(\tau);t_b)$.  It is given by the following Bellman equation: 
\begin{equation}\label{v0}
    \begin{aligned}
    V^0(\mu(\tau);t_b)= &\max\{\   \mu(\tau) v-t_b, \ 0,\\
   & - k d\tau +\mu(\tau) \lambda d\tau (v-t_b)  +(1-\mu(\tau)\lambda d\tau)V^0(\mu(\tau+d\tau);t_b) \}.
    \end{aligned}
\end{equation}
We interpret the equation as follows. At any time $\tau$ before purchase, the buyer can either purchase the product, walk away, or continue to learn. If he continues to learn for an increment of time $d\tau$, then, with probability $\mu(\tau)\lambda d\tau $, good news arrives and he purchases the item; 
with the remaining probability, no news arrives and his belief decreases to $\mu(\tau+d\tau)$. 
Note that his option value at the prior belief, $V^0(\mu_0; t_b)$, is the minimum payoff he can obtain from participation.

The seller's choice of return policy $(x_r,t_r)$ affects the buyer's learning behavior as follows. Given a  mechanism $m=\{t_b,(x_r,t_r)\}$, 
let $V_P(\mu(\tau);m)$ be the buyer's continuation value from learning   if he purchases the product first and learns afterward: 
\begin{equation*}\label{vp}
    \begin{aligned}
    V_P(\mu(\tau); m)=&\max\{\   \mu(\tau) v-t_b, \ \mu(\tau) v x_r-t_r,\\
   & - k d\tau +\mu(\tau) \lambda d\tau (v-t_b)  +(1-\mu(\tau)\lambda d\tau)V_P(\mu(\tau+d\tau);\ m) \}.
    \end{aligned}
\end{equation*}
This equation captures the fact that when the buyer stops learning after purchase, he can either request a return or decide to keep the product without requesting a return. 
Individual rationality requires $V_P(\mu_0; m)\ge V^0(\mu_0; m)$. Otherwise, the return policy would not be effective, because the buyer would learn before purchase and thereby guarantee his option value. Therefore, without loss of generality,  if the seller allows returns, we assume that the buyer learns after purchase.\footnote{ We adopt this assumption because it is more natural in practice. In the case that $V_P(\mu_0; m)= V^0(\mu_0; m)$, one can also assume the buyer learns all he needs to before purchase, and upon purchase, he immediately decides whether to request a return or keep the product.   }

\subsection{The buyer's endogenous option value}
In this  subsection, we discuss how the buyer's option value $V^0(\mu_0;t_b)$ varies with the price $t_b$. 
Conditional on learning, the Bellman equation (\ref{v0}) leads to the differential equation
\begin{equation}\label{odeb}\tag{ODE}
        (1-\mu)\mu \lambda V_1 (\mu;t_b)+ \mu \lambda V(\mu;t_b)=\mu \lambda (v-t_b)- k,
    \end{equation}
where $V_1(\mu; t_b)$ denotes the partial derivative 
with respect to
the first argument. For a fixed $t_b$, the buyer's optimal learning strategy is determined by two cutoff beliefs: the quitting belief $q(t_b)$ and the trial belief $Q(t_b)$, with $q(t_b)\le Q(t_b)$. Specifically, he continues to learn as long as his belief $\mu$ falls 
between the two cutoffs. 
The quitting belief $q(t_b)$ is determined by the standard value-matching and smooth-pasting conditions,
and it admits a closed form:
\begin{equation}\label{quitbelief}
    q(t_b)=\frac{k}{\lambda (v-t_b)}.
\end{equation}
The trial belief is the value of belief above which the buyer strictly prefers immediate consumption over further learning: 
\begin{equation}\label{bq}
    Q(t_b)=\{\mu:  V(\mu;t_b)=\mu v-t_b\}.
\end{equation}
In equation (\ref{bq}) and in the remainder of the paper, with slight abuse of notation, we use $V(\mu;t_b)$ to denote the solution of (\ref{odeb}) with boundary point $(q(t_b),0)$.\footnote{The formula for $V(\mu;t_b)$ can be found at the end of Appendix B. }



\begin{prop}\label{benchprop} 
There exist $\underline{t}_b<\overline{t}_b$ such that the following hold:
\begin{enumerate}\setlength\itemsep{0cm}
    \item   If $t_b\notin [\underline{t}_b, \overline{t}_b]$\ , 
    \[V^0(\mu;t_b)=\max\{0,\mu v-t_b\}.\] 
    \item If $t_b\in [\underline{t}_b, \overline{t}_b]$,
\begin{equation*}
V^0(\mu;t_b)=\left\{
             \begin{array}{lr}
            0,                                      &  \mu<q(t_b),\\
            V(\mu;t_b),  & q(t_b)\leq \mu<Q(t_b),\\
            \mu v- t_b,                 &  \mu\geq Q(t_b).
             \end{array}
\right.
\end{equation*}
\end{enumerate}
\end{prop}

Proposition \ref{benchprop} characterizes the buyer's continuation value from learning (as) if the mechanism disallows returns. His learning strategy is constructed via the quitting belief and the trial belief. As depicted in Figure \ref{fig:olive}(a), 
 both beliefs are increasing in price, and they coincide at the two boundaries.\footnote{
At higher prices, the high-type surplus is smaller, so the buyer's optimal learning period is shorter; i.e., the quitting belief is increasing in $t_b$.  The trial belief is also increasing in $t_b$, because higher prices make immediate consumption more attractive for a larger range of beliefs. 
} 
Clearly, when the price is sufficiently low or high (case 1 of the proposition), learning is suboptimal regardless of the buyer's belief.
When the price is moderate (case 2), if the buyer's prior belief lies in $[q(t_b), Q(t_b))$, then it is optimal for him to learn until either good news arrives (at which point he purchases the item), or his belief has fallen to $q(t_b)$ without good news arriving (at which point he walks away).  Let $\underline{\mu}=q(\underline{t}_b)=Q(\underline{t}_b)$ and $\overline{\mu}=q(\bar{t}_b)=Q(\bar{t}_b)$. Since learning is never optimal if the prior belief lies outside $[\underline{\mu},\overline{\mu}]$, from now on we assume $\mu_0\in [\underline{\mu},\overline{\mu}]$ unless otherwise specified. From Figure \ref{fig:olive}(a), it is obvious that for a given prior belief $\mu_0$, the buyer is willing to learn if the price $t_b$ lies in between $Q^{-1}(\mu_0)$ and $q^{-1}(\mu_0)$. 

         \begin{figure}[t]
        \centering
        \includegraphics[width=15cm]{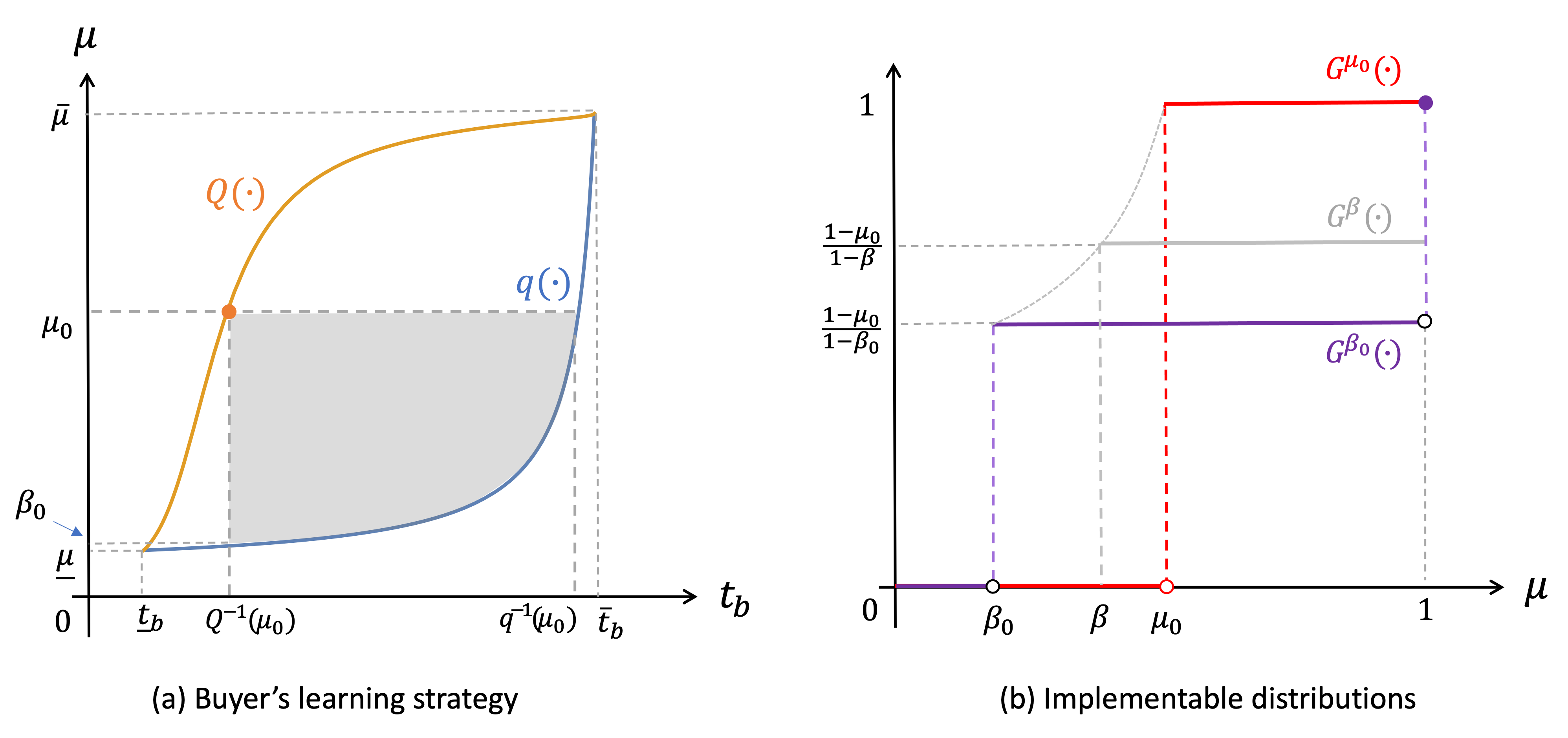}
        \caption{Buyer's learning strategy and distribution of stopping beliefs.}
        \label{fig:olive}
         \end{figure} 

\subsection{Implementable distributions}
Imagine that the buyer chooses the learning strategy such that he continues to learn until his posterior belief reaches either some value $\beta\le \mu_0$ or $1$. Ex ante, such a learning strategy corresponds to a certain distribution $G^{\beta}: [0,1]\rightarrow [0,1]$ of posterior (stopping) beliefs, constructed as follows: 
         \begin{equation*}
		  G^{\beta}(\mu) = \begin{cases}
			0 &\text{if } \mu \in [0,\beta), \\
			\frac{1-\mu_0}{1-\beta} &\text{if } \mu \in [\beta,1),  \\
			1 &\text{if } \mu=1.
		\end{cases}
		\label{G}
	\end{equation*}
 If $\beta \neq \mu_0$, then $G^{\beta}$ has a binary support $\{\beta,1\}$.
 Note that a smaller stopping belief $\beta$ indicates a longer (and hence more informative) learning process. 
Since $Q^{-1}(\mu_0)$ is the lowest price at which the buyer is willing to learn, the smallest stopping belief that can be realized is the quitting belief at that price, namely, $q(Q^{-1}(\mu_0))$. We let $\beta_0= q(Q^{-1}(\mu_0))$.\footnote{Note that $\beta_0$ is a function of $\mu_0$; however, for brevity, we write simply $\beta_0$ wherever there is no risk of ambiguity. } Hence, $G^{\beta_0}$ represents the most informative learning process possible. 
 
 If $\beta=\mu_0$, then $G^{\mu_0}$ represents the Dirac measure on the $\mu_0$.  It corresponds to the least informative learning process, under which the buyer does not learn at all. 
Figure  \ref{fig:olive}(b) shows $G^{\beta_0}$ in purple, $G^{\mu_0}$ in red, and one example of $G^{\beta}$ in grey.

Let $\mathcal{G}^{\mu_0}=\{G^{\beta}\}_{\beta\in[\beta_0,\mu_0]}$. Each element in this set 
corresponds to a reduced-form representation of the buyer's learning strategy. 

\begin{lem}\label{set}
    The set $\mathcal{G}^{\mu_0}$ consists of all distributions $G^\beta$ that can be implemented under some refund mechanism. 
\end{lem}
Lemma \ref{set} implies that if $\beta \notin [\beta_0,\mu_0]$, then $G^{\beta}$ is not an optimal learning strategy for the buyer under any refund mechanism.

\subsection{Implementable mechanisms}\label{sec:inplementable}
Since the buyer's learning strategy is the best response to the refund mechanism, we now study implementable mechanisms---those that can induce the buyer  to adopt a specific learning strategy. 

\begin{defn}
 A refund mechanism $m$ is $G^{\beta}$-\textit{implementable}, for some $\beta\in [\beta_0,\mu_0)$, if the following conditions hold:
\vspace{-0.1cm}
\begin{enumerate}[itemindent=-2em]\setlength\itemsep{0.1cm}
    \item[] (IR)~~$V_P(\mu_0;m)\ge V^0(\mu_0;t_b)$;
    \item[] (OB).
    $V_P(\mu_0;m)\ge \max \{\mu_0 v-t_b, \mu_0 v x_r-t_r\}$ ; $\beta=\inf \{\mu\in[0,1]: V_P(\mu;m)>\mu v x_r-t_r\}$;
     \item[] (IC)~~$\beta v x_r-t_r>\beta v-t_b$ and $v-t_b>v x_r-t_r$. 
\end{enumerate}\vspace{-0.2cm}
A refund mechanism $m$ is $G^{\mu_0}$-\textit{implementable} if the following conditions hold:
\vspace{-0.1cm}
\begin{enumerate}[itemindent=-2em]\setlength\itemsep{0.1cm}
\item[] (IR)~~$V_P(\mu_0;m)\ge V^0(\mu_0;t_b)$;
    \item[]  (OB)~~$V_P(\mu_0;m)= \mu_0 v-t_b.$
\end{enumerate}
\end{defn}

In the case $\beta\neq \mu_0$, the obedience constraint (OB) for a mechanism to be $G^{\beta}$-implementable  consists of two parts: the buyer must be willing to learn at the prior belief, and he must be willing to stop learning if his posterior belief reaches either $\beta$ or $1$. Incentive compatibility (IC) requires that when the buyer stops learning (at belief $\beta$ or $1$), he is willing to truthfully report his posterior belief by choosing the correct allocation rule, i.e., by requesting a return at belief $\beta$ and consuming (i.e., deciding not to request a return) at belief $1$. Interestingly, individual rationality (IR) and the obedience constraint (OB)  together imply IC. If, at stopping, the buyer preferred the same allocation rule regardless of his belief, then learning would not be necessary. For a $G^{\mu_0}$-implementable mechanism, the obedience constraint (OB) requires the buyer to be willing to consume without learning. 

For a given $\beta$, there could exist multiple $G^{\beta}$-implementable mechanisms. To simplify the problem, first note that IR binds under any optimal refund mechanism. 
\begin{lem}\label{newlem1}
   Under any optimal refund mechanism $m$, IR binds, i.e., $V_P(\mu_0;m)= V^0(\mu_0;t_b)$. 
\end{lem}
That is, under any optimal refund mechanism, the buyer obtains the same continuation value as he would if the mechanism prohibited returns.\footnote{To see this, suppose the mechanism is $G^{\beta}$-implementable and IR is slack.  Then the seller can increase the return transfer $t_r$ and adjust the return allocation probability $x_r$ accordingly without violating any of the conditions for $G^{\beta}$-implementability. }
This lemma implies that if $\beta\neq \mu_0$, then there is only one $G^{\beta}$-implementable mechanism for each fixed price. 

Before explicitly characterizing the implementable mechanisms, we define the buyer's marginal cost of information. Recall that, conditional on learning and as long as no news arrives, the buyer's posterior belief is a deterministic function of time $\mu(\tau)$, determined by the law of motion specified in equation (\ref{lawofmotion}).\footnote{We provide its characterization in Appendix B. }  With slight abuse of notation, let $\tau(\mu)=(\mu)^{-1}(\mu)$ be inverse of $\mu(\tau)$, which represents the time it takes the buyer to reach the belief $\mu$.  
Then $k|\tau'(\mu)|$ represents the marginal time cost per unit change in stopping belief.\footnote{Note that $|\tau'(\mu)|=1/|\mu'(\tau)|=1/(\lambda \mu(1-\mu))$. }   
We call this quantity the \textit{marginal cost of information} and denote it by $MC(\mu)$:
\[
MC(\mu):=k|\tau'(\mu)|=\frac{k}{\lambda \mu(1-\mu)}. 
\]

Define the return policy
$(x_r, t_r): [q(t_b),\mu_0)\times[0,v] \rightarrow [0,1]\times [0,t_b]$
as a mapping taking the buyer's privately reported stopping belief and the price to an allocation rule. 

\begin{prop}\label{implementable}
    A refund mechanism with price $t_b$ and return policy $(x_r(\beta,t_b),t_r(\beta,t_b))$ is $G^{\beta}$-implementable, where $\beta\in[q(t_b),\mu_0)$, if 
    \begin{equation}\label{smooth}
     x_r(\beta,t_b)=V^0_1(\beta;t_b)/v,
 \end{equation}
  \begin{equation}\label{value}
     t_r(\beta,t_b)=\beta V_1^0(\beta;t_b)-V^0(\beta;t_b)=\int_{q(t_b)}^{\beta} MC(\mu)\  d\mu.
 \end{equation}
  A  mechanism with price $t_b\le Q^{-1}(\mu_0)$ and a no-return policy $(1,t_b)$ is $G^{\mu_0}$-implementable. 
\end{prop}

Note that for a price $t_b$, $G^{q(t_b)}$ represents the longest (most informative) learning process possible: when the buyer stops learning at the quitting belief $q(t_b)$, he obtains zero continuation value. Therefore, to implement a shorter (less informative) learning process, with a stopping belief $\beta > q(t_b)$, the seller must offer the buyer a positive probability of keeping the product upon return. This  creates a positive allocation surplus upon return, so that the buyer prefers requesting a return over continuing to learn.
The implementability of $G^{\beta}$ requires the allocation rate $x_r$ to be proportional to the slope of the buyer's continuation value at $\beta$ (this is known as the smooth-pasting condition). Additionally, the return transfer $t_r$ must make the buyer indifferent between requesting a return and continuing to learn  (this is known as the value-matching condition). Therefore, $t_r$ is the difference between the allocation surplus upon return and the buyer's continuation value at belief $\beta$. This is the first equality in (\ref{value}).

Moreover, there is a deeper connection between the return transfer $t_r$ and the buyer's information cost. 
This is captured by the second equality in (\ref{value}), which explains how the seller extracts surplus from a return policy. 
If the buyer learns until his belief reaches the quitting belief, $\beta=q(t_b)$, he incurs the maximal information cost of $\int_{q(t_b)}^{\mu_0} MC(\mu) d\mu$. In this case, $t_r=\int_{q(t_b)}^{q(t_b)} MC(\mu) d\mu=0$; that is, the seller cannot generate a positive return transfer if the buyer incurs the maximal information cost. 
Conversely, if the buyer stops learning at some belief $\beta>q(t_b)$, he can lower his information cost by the amount of $\int_{q(t_b)}^{\beta} MC(\mu) d\mu$. However, these savings are entirely extracted by the seller through the return policy defined in Proposition \ref{implementable}.

We now define the \emph{transfer function} $t:[0,1]\times[0,v] \rightarrow [0,v]$, which maps the pair consisting of the buyer's reported stopping belief and the price to a transfer: 
 \begin{equation}
		  t(\mu;t_b) = \begin{cases}
			0 &\text{if } \mu \in [0,q(t_b)) ,\\
			t_r(\mu,t_b) &\text{if } \mu \in [q(t_b),\mu_0) , \\
			t_b &\text{if } \mu \in [\mu_0,1].
		\end{cases}
		\label{tmu}
	\end{equation}
Figure \ref{transferfunction}(a) shows a graph of the transfer function; note that it has a jump at $\mu_0$. 

We are now ready to write down the seller's program. We distinguish between two cases: (a) learning deterrence, which implements $G^{\mu_0}$, the Dirac measure on the prior belief, and (b) learning encouragement, which implements $G^{\beta}\in \mathcal{G}^{\mu_0}\setminus G^{\mu_0}$ (i.e., implementable distributions other than $G^{\mu_0}$). 

\textbf{Learning deterrence:} In this case, the seller's program is
\begin{equation}\label{D}\tag{D}
\begin{aligned}
  \sup_{t_b} &\ \int_0^1 t(\mu;t_b) \ dG^{\mu_0} (\mu)      \\
  &\textnormal{s.t. } \ \ t_b\le Q^{-1}(\mu_0) .
\end{aligned}
\end{equation}
Since $G^{\mu_0}$ assigns probability one to the prior belief, the optimization above maximizes the price, subject to the condition that $G^{\mu_0}$ is implementable, i.e., $t_b\le Q^{-1}(\mu_0)$.
Obviously, $t_b=Q^{-1}(\mu_0)$ is the highest price under which the buyer will forgo private learning. 
We define \textit{learning deterrence} as a mechanism with the learning-deterrence  price $t_b= Q^{-1}(\mu_0)$ and a no-return policy.  Such a mechanism is always feasible, so the revenue from it is a lower bound on the revenue from the optimal refund mechanism.

\textbf{Learning encouragement:} In this case, the seller's program is
\begin{equation}\label{P}\tag{P}
\begin{aligned}
     \sup_{t_b\in (Q^{-1}(\mu_0),q^{-1}(\mu_0)]} &\left\{\quad \sup_{G^{\beta}}\ \int_0^1 t(\mu;t_b) \ dG^{\beta} (\mu) \quad \right\}\\
     &\qquad\textnormal{s.t. } \ \ \ G^{\beta} \in \mathcal{G}^{\mu_0}|t_b,
\end{aligned}
\end{equation}
where $\mathcal{G}^{\mu_0}|t_b=\{G^{\beta}\in \mathcal{G}^{\mu_0}: \beta\in [q(t_b),\mu_0) \}$ is the set of implementable distributions under price $t_b$, excluding  $G^{\mu_0}$.  

The inner optimization of (\ref{P}) is a typical information design problem, in which the seller maximizes the expected value of a function (in this case the transfer function) over a set of distributions (in this case the implementable distributions).\footnote{In standard information design, the information designer's problem is $\max_{G} \int t dG$ subject to the constraint that $G$ is a mean-preserving spread of the prior belief $\mu_0$ (see  \cite{KamenicaGentzkow2011}).  Here, we restrict the maximization to the set of implementable distributions, and the transfer function is constructed through implementability.  } 
Since the inner optimization is carried out for a fixed price, in the second step (the outer optimization) the seller  optimizes with respect to the price. 
Note that in the outer optimization, we consider only prices higher than the learning-deterrence price $Q^{-1}(\mu_0)$.
This is because the seller's expected revenue from implementing $G^{\beta}$ with the learning-deterrence price is strictly below her revenue from the learning-deterrence mechanism (i.e., $Q^{-1}(\mu_0)$), even as $\beta\rightarrow \mu_0$.\footnote{Recall that the transfer function has a jump at $\mu_0$. } 
Thus, at the learning-deterrence price, encouraging learning is always suboptimal.

In Section \ref{Exogenousprice}, we study the optimal learning process $G^{\beta}$ for various exogenously fixed prices. 
In Section \ref{endogenousprice} and \ref{opm}, we endogenize the price and characterize the optimal refund mechanism, which is the solution to the maximum of the two programs (\ref{D}) and (\ref{P}).

\section{Learning processes for exogenously fixed prices}\label{Exogenousprice}

In this section we discuss how the optimal learning process $G^{\beta}$ changes with the price and show how the seller might benefit from implementing an interior stopping belief $\beta>q(t_b)$. 
Let $\beta^*(t_b)$ be the optimal stopping belief for a given price $t_b$. 
Recall that for a given prior belief  $\mu_0$, the relevant range of prices is $[Q^{-1}(\mu_0),q^{-1}(\mu_0)]$. 
Obviously, with the learning-deterrence price, $t_b=Q^{-1}(\mu_0)$, the optimal learning process is the one with no learning, $G^{\mu_0}$. 
For $t_b>Q^{-1}(\mu_0)$, the optimal learning process is the solution to the inner optimization of (\ref{P}) across all implementable distributions (i.e.,  all $G^{\beta}$ with $\beta\in [q(t_b),\mu_0)$).

\begin{figure}[t]
        \centering
        \includegraphics[width=15cm]{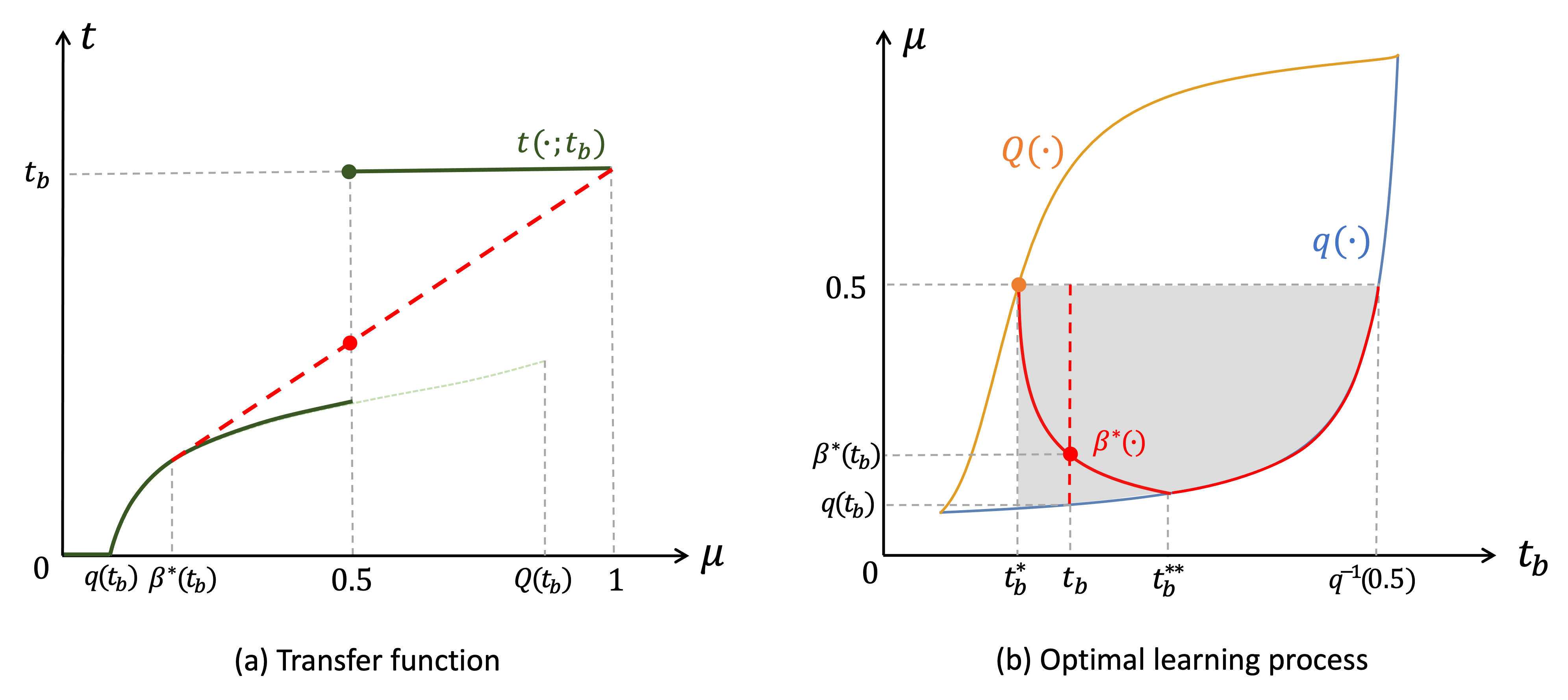}
        \caption{Optimal information process for a fixed price. \small In (b), the orange dot and the red curve represent the optimal stopping belief $\beta^{*}$ as a function of $t_b$. When $t_b\in(t^*_b,t^{**}_b)$, the solution is interior, i.e., $\beta^{*}(t_b)=\beta^o(t_b)$. 
        }
        \label{transferfunction}
         \end{figure} 

For some prices, an intermediate (interior) learning process may be optimal. To see why, note that the transfer function $t(\cdot;t_b)$ given by (\ref{tmu}) coincides with the return transfer $t_r(\cdot,t_b)$ when $\mu\in[q(t_b),\mu_0)$.  Hence, when $\mu\in [q(t_b),\mu_0)$, the slope of the transfer function equals the marginal cost of information: 
\[
t'(\mu;t_b)=t_r'(\mu;t_b)=MC(\mu)=\frac{k}{\lambda \mu(1-\mu)}.
\]
That is, if the buyer acquires one more unit of information (as measured by the decrease in his stopping belief), the seller's return transfer falls by a margin of $MC(\mu)$. Notably, the marginal cost of information exhibits symmetry over beliefs, i.e.,  $MC(\mu)$ is symmetric about $\mu=0.5$, and it is higher when the belief is more extreme.\footnote{It indicates that it takes longer for the belief to move from $\mu$ to $\mu-d$ for small $d$ if $\mu$ is close to 0 or 1. }
Consequently,  the marginal loss on the return transfer is greater when the buyer's stopping belief is closer to the quitting belief (if the quitting belief is small). This could plausibly imply that an intermediate level of learning is optimal.

For the sake of exposition, consider the special case where the prior belief is $0.5$.
Since the marginal cost of information is decreasing on $[0,0.5]$, the transfer function $t$ is increasing and concave on $[q(t_b),0.5)$. 
 Figure \ref{transferfunction}(a) presents a transfer function with price $t_b$ and $\mu_0=0.5$. By a slight modification of the concavification argument in information design,\footnote{The concavification is modified as follows: it is given by $\sup_{\beta, \gamma \in[0,1]} \gamma t(\beta;t_b)+(1-\gamma) t(1;t_b) $, subject to Bayesian plausibility  $\gamma \beta+(1-\gamma) 1=\mu_0$ and an additional restriction on the distributions, namely $\beta\in[q(t_b),\mu_0)$.  
 The red dashed line in Figure \ref{transferfunction}(a) is the concavification on $\mu_0\in[\beta^*(t_b),1]$ where $\beta^*(t_b)$ is the solution to the concavification for those values of $\mu_0$. } 
 we have that the red dot in Figure \ref{transferfunction}(a) indicates the maximum expected value of the transfer function across the set of implementable distributions under the price $t_b$.\footnote{All distributions with a binary support $\{\beta,1\}$ where $\beta \in [q(t_b),0.5)$.}
 It is a convex combination of $t(\beta^*(t_b);t_b)$ and $t(1;t_b)$ with respect to the distribution $G^{\beta^*(t_b)}$. 
The figure thus shows that $G^{\beta^*(t_b)}$ induces a local maximum. Hence, if $ \beta^*(t_b)\in(q(t_b),\mu_0)$ is an interior stopping belief, then it is determined by the fact that the marginal revenue of information equals zero,
\[
  \lim_{d \rightarrow 0}\  \frac{\int_0^1 t d G^{\beta+d}-\int_0^1 t dG^{\beta}}{d}=0 . 
\]
 Let $\gamma(\beta):= G^{\beta}(\beta)$ be the probability that the belief $\beta$ is realized; it is also the probability that the buyer will request a return, so we sometimes refer to it as the \emph{return rate}. The above limit condition reduces to the first-order condition 
\begin{equation}\label{foc}
     \frac{t'(\beta;t_b)}{t(1;t_b)-t(\beta;t_b)}=\frac{\gamma'(\beta)}{\gamma(\beta)}\equiv \frac{1}{1-\beta}.
 \end{equation}
 The left-hand side is the proportional change in transfer, while the right-hand side is the proportional change in return rate. 
Let $\beta^{o}(t_b)$ be the solution to equation (\ref{foc}). It is decreasing in $t_b$,\footnote{Implicit differentiation of equation (\ref{foc}) with respect to $t_b$ yields $(\beta^o)'=\frac{(\beta^o)^2 (1-\beta^o)}{(2 \beta^o-1)(1-q(t_b))k/\lambda}\le 0$. \label{footnote} } indicating that when the price is increased, the seller leverages a stochastic return policy to prolong the buyer's learning process and make a successful sale more likely. In other words, with the optimally designed stochastic return policy, the seller can simultaneously increase the price and the probability of a successful sale.

The monotonicity of $\beta^{o}(t_b)$ implies that there exist two thresholds $t^*_b<t^{**}_b$, plotted in Figure \ref{transferfunction}(b), such that
\begin{equation}\label{boundary}
    \beta^o(t^*_b)=Q(t^*_b)\quad \text{and} \quad \beta^o(t^{**}_b)=q(t^{**}_b).
\end{equation}
Substituting them in the first-order condition (\ref{foc}), we obtain $t^*_b=Q^{-1}(0.5)$ and $t^{**}_b=v/2$,\footnote{A detailed calculation can be found in the proof of Proposition \ref{step1} in Appendix A. 
}
both of which are invariant to the prior belief.\footnote{This is because  $\beta^o(t_b)$ is invariant to the prior belief by the typical concavification argument: the support of the optimal information structure does not depend on the prior if the prior belongs to a region where the concavification is strictly above the transfer (value) function.} 
In particular, $t^*_b$ happens to be the learning-deterrence price for $\mu_0=0.5$. 

We are now ready to summarize  the optimal learning process and the optimal return policy for an exogenously fixed price. 
\begin{prop}\label{step1}
There exist two thresholds, $t^*_b<t^{**}_b$, such that for $\mu_0=0.5$ the optimal learning process $G^{\beta^*}$ is as follows:
\begin{enumerate}\setlength\itemsep{0cm}
 \item If $t_b=t^*_b$, then   $\beta^*(t_b)=\mu_0$ (no learning).
    \item If $t_b\in (t^*_b,t^{**}_b)$, then   $\beta^*(t_b)=\beta^o(t_b)$ (interior learning).
    \item If $t_b\in [t^{**}_b,q^{-1}(\mu_0)]$, then  $\beta^*(t_b)=q(t_b)$ (full learning). 
\end{enumerate}
The optimal return policy for the price $t_b$ is determined by the implementation of $G^{\beta^*(t_b)}$. 
\end{prop}

The orange dot and the red curve in Figure \ref{transferfunction}(b) represent the optimal stopping belief $\beta^*(t_b)$ when $\mu_0=0.5$. 
The orange dot corresponds to case 1 of Proposition \ref{step1}: at the learning-deterrence price $t_b=t^*_b=Q^{-1}(0.5)$,   a no-return policy is optimal. 
The decreasing part of the red curve corresponds to case 2: at a moderate price $t_b\in (t^*_b,t^{**}_b)$, the optimal return policy is stochastic and is determined by the implementation of $G^{\beta^o(t_b)}$. In this case, the buyer's stopping belief under the optimal return policy is decreasing in price. 
Finally, the increasing part of the red curve corresponds to case 3: at a price higher than $t^{**}_b$, a free-return policy is optimal. In this case, the buyer's stopping belief under the optimal return policy is increasing in price.

For $\mu_0\neq 0.5$, the structure of the optimal learning process is similar to that described in Proposition \ref{step1}, but with additional considerations related to the relevant ranges of prices. 
We provide a complete statement in Corollary \ref{fullcharacterization} in Appendix A.

\section{The optimal return policies}\label{endogenousprice}
In this section, we study the seller's optimization over prices, i.e., the outer optimization of (\ref{P}). We restrict our attention to the case $t_b\in(t^*_b,t^{**}_b)$, because this is the range of prices for which the optimal refund mechanism could involve a stochastic return policy. In this case, the induced learning process $G^{\beta}$ must satisfy $\beta=\beta^o(t_b)$, and so (\ref{P}) can be simplified to the following:  
\begin{equation*}\label{PR}\tag{P-R}
\begin{aligned}
     \sup_{t_b\in(Q^{-1}(\mu_0),q^{-1}(\mu_0)]}   \quad &\int_0^1 t(\mu;t_b) d G^{\beta^o(t_b)}(\mu)\\
    \textnormal{s.t. }\qquad \quad & \qquad t^*_b\le t_b \le t^{**}_b.
\end{aligned}
\end{equation*}
If the solution to this problem is interior with respect to both constraints on $t_b$, then the optimal return policy may be stochastic. If the solution lies at the boundary of either constraint, then the optimal return policy has to be deterministic.  

To study this problem, we first define the 
\textit{price elasticity of the return rate}:
\[
\mathcal{E}_r(\beta):=\frac{d\gamma(\beta)/ \gamma(\beta)}{d t_b/t_b}.
\]
Recall that $\gamma(\beta)$ denotes the return rate.
The elasticity $\mathcal{E}_r(\beta)$ measures how sensitive the return rate is to the price; it depends on the buyer's learning process and how it changes with the price. 
The next lemma shows that when the learning process is optimally adjusted to the price, i.e., $\beta=\beta^o(t_b)$, the price elasticity of the return rate is larger (in absolute value) than it is when the buyer's learning process is not manipulated (i.e., when $\beta=q(t_b)$). 

\begin{lem}\label{elasticity}
    We have $|\mathcal{E}_r(\beta^o(t_b))|>\mathcal{E}_r(q(t_b)) $. 
\end{lem}

Intuitively, suppose we view $1-\gamma(\beta)$, the probability of a successful sale, as the demand. Then, by the law of demand, a higher price should imply lower  demand and thus a higher return rate; hence the price elasticity of the return rate should be positive.  This is indeed the case under full learning, when the buyer stops learning at the quitting belief $q(t_b)$. 
However, when the seller uses a stochastic return policy to optimize the buyer's learning process as in Section \ref{Exogenousprice}, an increase in price can decrease the return rate. Hence the price elasticity of the return rate becomes negative, creating an exception to the law of demand. In this case, higher elasticity (in absolute value) is actually beneficial, because the demand moves in the same direction as the price.

\begin{thm}\label{deterministic}
   The optimal mechanism involves a deterministic return policy, i.e., $x_r\in\{0,1\}$. 
\end{thm}
\begin{proof}
    
Since $G^{\beta^o(t_b)}$ has binary support $\{\beta^o(t_b),1\}$, the expected revenue is a convex combination of the return transfer and the price,
\[\gamma(\beta^o) \times t_r(\beta^o;t_b)+(1-\gamma(\beta^o)) \times t_b. 
\]
For simplicity, we write $\beta^o$ for $\beta^o(t_b)$. 
Given that $\beta^o(t_b)$ is determined via the first-order condition (\ref{foc}), the envelope theorem implies that the derivative of the expected revenue with respect to the price is
\[
\gamma(\beta^o )\frac{\partial t_r(\beta^o ;t_b)}{\partial t_b} + 1-\gamma(\beta^o ),
\]
where $\frac{\partial t_r(\beta^o ;t_b)}{\partial t_b}=-\frac{q(t_b)}{1-q(t_b)}$, since the price affects the return transfer $t_r$ only through the starting point of the integration. 
Hence, the marginal revenue from an increase in price is positive if and only if 
 \begin{equation}\label{focprice}
     \frac{1-\gamma(\beta^o )}{\gamma(\beta^o)}-\frac{q(t_b)}{1-q(t_b)}\ge 0. 
 \end{equation}

 It turns out that the left-hand side of the above inequality, as a function of $t_b$, single-crosses 0 from below (if it ever crosses 0 at all). To see this, 
 let (\ref{focprice}) hold with equality; then we obtain a critical price at which
the slope of the left-hand side equals\footnote{Take the derivative of the left-hand side with respect to $t_b$ and multiply it by $\gamma(\beta^o)$; then, conditional on equality in (\ref{focprice}), we obtain
\[
 \frac{-\gamma'(\beta^o) (\beta^o)'}{\gamma(\beta^o)}-\frac{q'(t_b)}{1-q(t_b)}= \frac{-\gamma'(\beta^o) (\beta^o)'}{\gamma(\beta^o)}-\frac{\gamma'(q) q'}{\gamma(q)}=\frac{-d\gamma(\beta^o)/d t_b}{\gamma(\beta^o)}-\frac{d\gamma(q)/d t_b}{\gamma(q)}.
\]
The first equality comes from $\frac{\gamma'(q)}{\gamma(q)}=\frac{1}{1-q}$, and the second comes from simple rearrangement.  Multiplying the above expressions by $t_b$, we obtain the elasticity. } 
\[
\frac{|\mathcal{E}_r(\beta^o(t_b))|-\mathcal{E}_r(q(t_b))}{\gamma(\beta^o) t_b} >0.
\]
The numerator is positive by Lemma \ref{elasticity}. Hence, the left-hand side of (\ref{focprice}) single-crosses 0 from below.  

 This means that the marginal revenue with respect to the change in price single-crosses 0 from below (if it ever crosses), which implies that the expected revenue is quasi-convex in price; hence the optimal price lies at the boundary of one of the constraints in (\ref{PR}). 
 For example, in the case $\mu_0=0.5$, the solution to (\ref{PR}) is either $t^*_b$ or $t^{**}_b$. By Proposition \ref{step1}, this shows that the optimal return policy is deterministic.  
\end{proof}

At first glance, it may be surprising that the optimal refund mechanism is always deterministic, especially because stochastic returns greatly 
enlarge the set of learning strategies that the seller can implement.  
The deeper reasoning here is that a strategically designed stochastic return policy can generate an exception to the law of demand. To elaborate, let us decompose the seller's expected revenue as follows:
\[
 t_r +(1-\gamma) \times (t_b- t_r).
\]
The first component represents the return transfer, which is the revenue the seller is guaranteed to obtain from either type of buyer. The second component is the additional revenue she can hope for if the high-type buyer receives good news. The seller can follow either a niche-marketing strategy (to maximize her additional revenue from the high-type buyer) or a mass-marketing strategy (to maximize the return transfer). If she follows the niche-marketing strategy, then she increases the price $t_b$ so as to gain more from the high-type buyer (note that $t_b- t_r(\beta^o(t_b);t_b)$ is increasing in $t_b$).
Simultaneously, by choosing the optimal stochastic return policy, she increases the probability of a successful sale. This amplifies the second component of her revenue and thus  reinforces her incentive to follow the niche-marketing strategy. 
On the other hand, if she follows the mass-marketing strategy, then she reduces the price to obtain a higher return transfer (note that the return transfer is decreasing in price). This amplifies the first component (as now the return transfer accounts for a larger proportion of the total revenue) and thus reinforces her incentive to follow the mass-marketing strategy. 
This explains why the seller's expected revenue is quasi-convex in price, conditional on the optimality of the return policy.

In our example with $\mu_0=0.5$, this quasi-convexity implies that the solution to (\ref{PR}) is either $t^*_b$, the learning-deterrence price, or $t^{**}_b$, a price high enough to render free returns optimal. That is, the seller's expected revenue from implementing any point along the interior of the curve $\beta^o(t_b)$ (the decreasing part of the red curve in Figure \ref{transferfunction}(b)) is smaller than her revenue at the endpoints of this curve. Nevertheless, $t^{**}_b$ is just one possible price for which a free-return policy is optimal; there could be another price $t_b\in [t^{**}_b,q^{-1}(0.5))$ that leads to an even higher expected revenue. 
To find the optimal refund mechanism, we must  compare the learning-deterrence revenue and the highest free-return revenue. The relationship between these two quantities depends on the buyer's prior belief. We discuss this topic in the next section.

\section{The optimal refund mechanism}\label{opm}

We now identify the optimal refund mechanism for various prior beliefs.

\begin{thm}\label{two} 
There is a closed interval $F \subset [\underline\mu,\overline\mu]$ such that the optimal refund mechanism takes the following form:
\vspace{0cm}
\begin{enumerate}\setlength{\itemsep}{0cm}
\item If $\mu_0 \notin[\underline\mu,\overline\mu]$, then the mechanism offers a no-return policy with price $\mu_0 v$.
\item If $\mu_0 \in [\underline\mu,\overline\mu]$ and $\mu_0 \notin F$, then the mechanism offers a no-return policy with the learning deterrence price $Q^{-1}(\mu_0)$.
\item If $\mu_0 \in F$, then the mechanism offers a free-return policy with price $t^{F}_b(\mu_0)$,
where 
\[
t^F_b(\mu_0):=\argmax_{t_b} \int_0^1 t(\mu;t_b) d G^{q(t_b)}(\mu).
\]
\end{enumerate}
\vspace{-0.3cm}
Moreover, there exists a number $c^{*}$ such that the following hold: if $\frac{k}{\lambda}\le c^{*}$,  then $F$ is non-empty and the optimal mechanism takes all three forms; if $\frac{k}{\lambda}>c^{*}$, then $F=\varnothing$ and the optimal mechanism takes the forms in items 1 and 2;  if $\frac{k}{\lambda}\rightarrow 0$, then $F=[0,1]$ and the optimal mechanism takes the form in item 3.  
\end{thm}

Note that when the optimal refund mechanism takes the first or second form, the buyer purchases the product immediately without learning. Hence the optimal mechanism maximizes both welfare and profits. By contrast, when the optimal mechanism takes the third form, the buyer learns until his belief reaches the corresponding quitting belief. Hence the probability that he will eventually return the product is positive. In this case, there is a conflict between profit maximization and welfare maximization. This conflict is triggered by the magnitude of the buyer's option value from learning.

To elaborate,  consider the second diagram in Figure \ref{opmgraph}, where the effective learning cost is not very high, i.e., $\frac{k}{\lambda}\le c^*$. 
As we saw in Proposition \ref{benchprop}, if the buyer's prior belief is very close to 0 or 1 (i.e., if $\mu_0\notin [\underline\mu,\overline\mu]$, as indicated by the two black intervals at the boundaries of the diagram), then the buyer is already well informed up front, and learning is not valuable to him. 
Consequently, his option value from learning is zero. Thus, if the seller offers a non-refundable price equal to $\mu_0 v$ (as in case 1 of Theorem \ref{two}), then the buyer will make the purchase, and the seller will obtain the entire first-best surplus.

The buyer's option value from learning becomes positive when learning becomes valuable to him, i.e., when $\mu_0 \in [\underline\mu,\overline\mu]$. Nevertheless, for $\mu_0$ close to $\underline\mu$ or $\overline\mu$,  learning is not that valuable. 
In such cases, the seller can deter learning and induce immediate consumption using only a small decrease in the price (deducted from the ex-ante expected valuation).
However, when the buyer's prior belief is more uncertain, learning is highly valuable, which makes learning deterrence considerably less profitable (as a significant price reduction is needed to compensate for the buyer's option value).
Alternatively, the seller may increase his expected revenue by setting a relatively high price and simultaneously allowing free returns. Although a free-return policy leads to welfare loss, as the buyer may eventually return the product, he is unlikely to do so if his prior belief is high.  
Consequently, the optimal refund mechanism allows free returns only if the buyer's prior belief is not too extreme and relatively optimistic (as indicated by the moderately rightward position of the blue interval in the second diagram of Figure \ref{opmgraph}).\footnote{This result is in contrast to \cite{pease2018}. There, buyer learning is modeled through Brownian motion, with high (low) valuation inducing an upward (downward) drift. There is a subtle difference between learning through Brownian motion and learning through a Poisson process. Suppose the seller sets the price to prevent buyer learning. Under Brownian learning, the buyer's option value from learning is 0 if and only if the prior belief is exactly 0 or 1. In contrast, under Poisson learning, his option value is 0 as long as the prior belief is close to 0 or 1. This is because, with Brownian motion, signals are unbounded; hence the value of learning is positive even for extreme beliefs close to 0 or 1. Under Poisson learning, on the other hand, learning is not valuable for extreme beliefs close to 0 or 1 (because the expected time needed to reach a decision is bounded below by $1 / \lambda$). Hence we obtain a  different prediction: learning deterrence is optimal for both sufficiently high and sufficiently low beliefs.
 }

\begin{figure}[t]
        \centering
        \includegraphics[width=13cm]{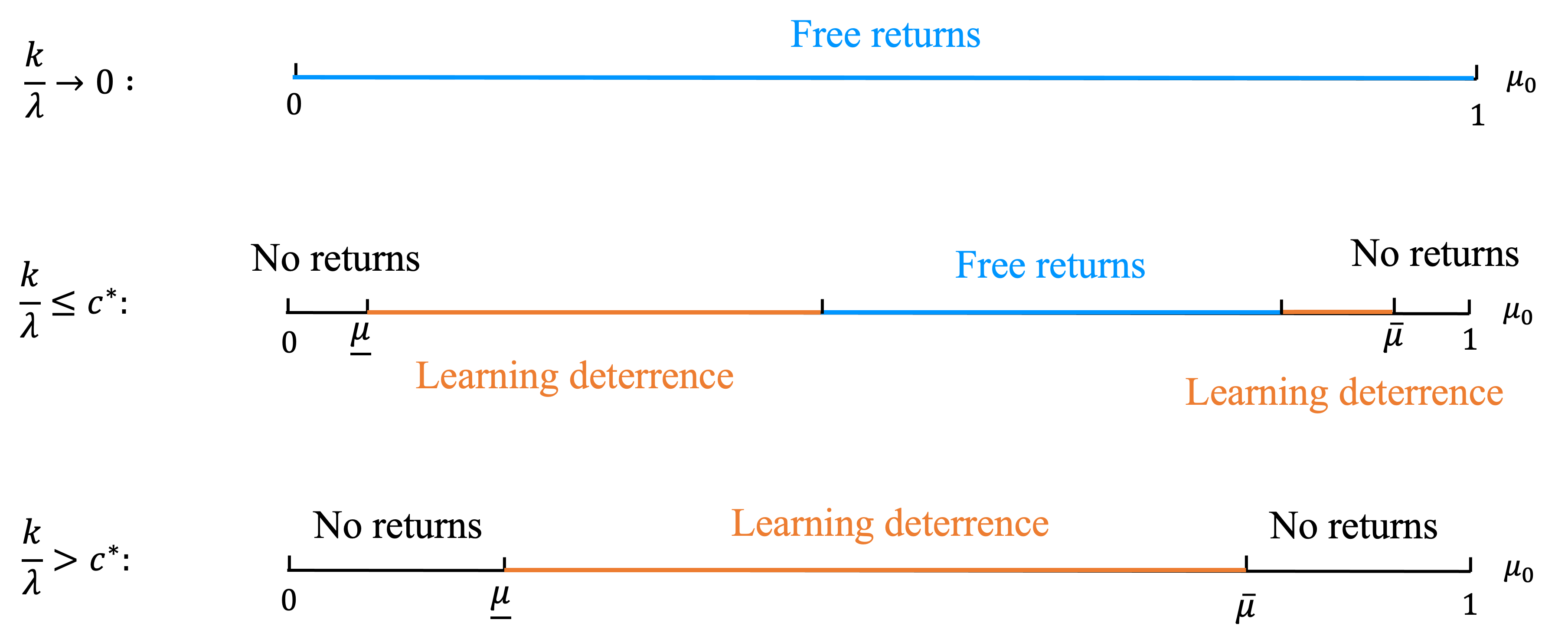}
        \caption{Optimal refund mechanism.}
        \label{opmgraph}
        \end{figure}

For sufficiently high learning costs (i.e., for $\frac{k}{\lambda}>c^*$), the buyer's option value from learning is minimal, and the optimal mechanism therefore takes only the first and second forms from Theorem \ref{two}. Thus, profit maximization and welfare maximization coincide across the entire range of prior beliefs. 
The third diagram of Figure \ref{opmgraph}) represents this case. 
On the other hand, as the learning cost converges to 0, the buyer acquires nearly perfect information. In this case, the seller's optimal strategy is to allow free returns, set the price equal to $v$, and sell exclusively to high-valuation buyers,  
as shown in the first diagram of Figure \ref{opmgraph}.\footnote{This outcome is the same as that of the robust refund mechanism of \cite{hinnosaar2020} when the restocking fee is zero. }

A free-return policy is always bundled with price discrimination. While the buyer benefits from being able to return the product, such a contract lowers his overall welfare, as the seller captures a major portion of the surplus if he decides to keep the product. Thus, learning costs have a non-monotonic effect on buyer welfare. For example, reducing the learning cost increases the buyer's option value from learning, thereby amplifying the seller's incentive to encourage learning. Paradoxically, this could lead to a scenario in which a slight decrease in the learning cost significantly reduces the buyer's welfare, as it leads the seller to permit returns and ultimately escalates price discrimination. 
\begin{figure}[t]
        \centering
        \includegraphics[width=18cm]{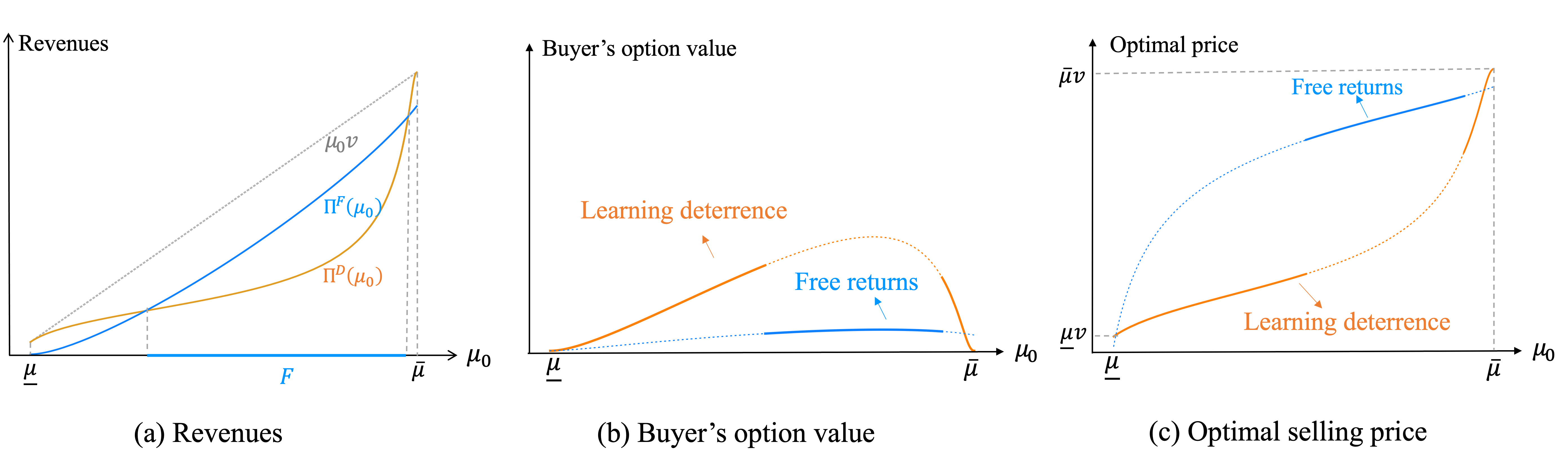}
        \caption{Comparisons between learning deterrence and free returns. \footnotesize 
        The graphs are constrained to $\mu_0\in[\underline\mu,\overline\mu]$, with learning deterrence and free returns depicted in orange and blue, respectively. Solid curves represent optimal mechanisms. Panel (a) shows the seller's revenue from the optimal mechanism, which is continuous in priors despite variations in the form of the optimal mechanism. In Panel (b), the buyer's option value (also interpreted as his participation value) is significantly reduced when the optimal mechanism allows for free returns. Panel (c) illustrates how the gap in the buyer's surplus can be attributed to the significantly higher selling price under a free-return mechanism.}
        \label{buyersurplus}
        \end{figure}

\section{More efficient post-purchase learning}\label{posteff}

If the learning process is more efficient after purchase, i.e., $\lambda_P>\lambda_B=\lambda$ (rather than $\lambda_P=\lambda_B=\lambda$ as we have been assuming so far), 
then the purchase itself generates extra information value. However, the seller can fully extract this extra information value by charging a cancellation fee to make the buyer indifferent between pre-purchase and post-purchase learning. 
Let $t_c$ denote the cancellation fee, i.e., the additional transfer the buyer has to pay for a return. A refund mechanism with a cancellation fee can be represented as $m:= \{t_b,(x_r,t_r+t_c)\}$. Hence, the mechanism space remains unchanged. However, the return transfer can now be decomposed into two components, each calculated through different constraints. 

As in the main model,  we use $V^0(\mu;t_b)$ to represent the buyer's option value from pre-purchase learning. We let $V_P(\mu;m)$ denote the buyer's continuation value from post-purchase learning with the post-purchase learning rate  $\lambda_P$.

\begin{cor}\label{post}
    Under the optimal mechanism $m$, the buyer's IR binds: $V_P(\mu_0;m)=V^0(\mu_0;t_b)$. 
\end{cor}

\begin{proof}
    This is proved in the same way as Lemma \ref{newlem1}. 
\end{proof}

We refer to a return policy $(x_r,t_r+t_c)$ as a \textit{return policy with cancellation fee} if $x_r=0$, $t_r=0$, and $t_c>0$. We use $m^c(t_b,t_c)$ to denote a mechanism with such a return policy. 
Under such a mechanism, the buyer conducts the maximum amount of learning; thus, his stopping belief in this case can be interpreted as the quitting belief for post-purchase learning.
Denoting the latter by $q_P(t_b,t_c)$, we have
\[
q_P(t_b,t_c)=\frac{k}{\lambda_P(v-t_b+t_c)}.
\]
Hence, $V_P(\mu_0;m^c)$ is the solution to (\ref{odeb}) with the boundary point $(q_P(t_b,t_c),-t_c)$ and the post-purchase learning rate $\lambda_P$.

 \begin{cor}\label{postimplementable} 
    A  mechanism with price $t_b$ and return policy $(x_r(\beta,t_b),t_r(\beta,t_b)+t_c(t_b))$ is $G^{\beta}$-implementable, for some $q_P(t_b,t_c(t_b))\le \beta<\mu_0$, if the following hold:
    \begin{enumerate}
        \item The cancellation fee $ t_c(t_b)$ is the solution to $V_P(\mu_0;m^c(t_b,t_c))=V^0(\mu_0;t_b)$. 
        \item We have $ x_r(\beta,t_b)=V'_P\left(\beta;m^c(t_b,t_c(t_b))\right)/v$ and $t_r(\beta,t_b)=\int_{q_P(t_b,t_c(t_b))}^{\beta} MC_P(\mu) d\mu.$
    \end{enumerate}
  A mechanism with price $t_b\le Q^{-1}(\mu_0)$ and a no-return policy $(1,t_b)$ is $G^{\mu_0}$-implementable. 
\end{cor}

\begin{proof}
    This is immediate from Corollary \ref{post} and Proposition \ref{implementable}. Note that $MC_P(\mu)$ denotes the marginal cost of information with the post-purchase learning rate $\lambda_P$. 
\end{proof}
The above corollary characterizes the $G^{\beta}$-implementable mechanisms. In the case $\beta\neq \mu_0$, the cancellation fee $t_c$ is determined by the (binding) IR constraint and is therefore constant with respect to the stopping belief $\beta$.
In contrast, the return transfer $t_r$ is determined by the obedience constraint (OB) and is therefore a function of $\beta$.
For $G^{\mu_0}$, the implementable mechanism is the same as in the main model. Corollary \ref{postimplementable} generalizes Proposition \ref{implementable}: if $\lambda_P=\lambda_B$, then the cancellation fee is zero. 

With a cancellation fee, we can construct a new  transfer function $t:[0,1]\times[0,v] \rightarrow [0,v]$ as follows:
 \begin{equation*}
		  t(\mu;t_b) = \begin{cases}
			0 &\text{if } \mu \in [0,q_P(t_b,t_c(t_b))), \\
			t_r(\mu;t_b)+t_c(t_b) &\text{if } \mu \in [q_P(t_b,t_c(t_b)),\mu_0),  \\
			t_b &\text{if } \mu \in [\mu_0,1].
		\end{cases}
	\end{equation*}
Then, to encourage learning, the seller's optimization problem is 
\begin{equation*} 
\begin{aligned}
     \sup_{t_b\in (Q^{-1}(\mu_0),q^{-1}(\mu_0;\lambda_P)]} &\left\{\qquad \sup_{G^{\beta}}\ \int_0^1 t(\mu;t_b) \ dG^{\beta} (\mu) \qquad \right\}\\
     &\qquad\ \textnormal{s.t. } \  q_P(t_b,t_c(t_b))\le \beta< \mu_0,
\end{aligned}
\end{equation*}
where  $q^{-1}(\mu_0;\lambda_P)$ is the inverse of the quitting belief $q(\cdot)$ defined in equation (\ref{quitbelief}) with $\lambda$ replaced by $\lambda_P$. This is the highest price at which the buyer is willing to learn.  The solution to this problem yields the best refund mechanism, conditional on the seller encouraging the buyer to learn. Comparing this supremum value with $Q^{-1}(\mu_0)$---the learning-deterrence revenue---leads to the optimal refund mechanism.

The next proposition characterizes the limit as $\lambda_P \rightarrow \infty$, the situation in which the buyer learns his true valuation almost immediately after purchase.

\begin{prop}\label{postopm}
 Suppose $\lambda_P\rightarrow \infty$.  The following mechanism is optimal:
\begin{enumerate}\setlength{\itemsep}{0cm}
\item If $\mu_0 \notin[\underline\mu,\overline\mu]$, then the mechanism offers a no-return policy with price $\mu_0 v$.
\item If $\mu_0 \in [\underline\mu,\overline\mu]$, then the mechanism offers a return policy with cancellation fee $\frac{k}{\lambda (1-\mu_0)}$ and price $q^{-1}(\mu_0)$. 
\end{enumerate}
\vspace{-0.3cm}
Under this mechanism, the seller's expected revenue is $\mu_0 v$, while the buyer's expected surplus from trade is 0. 
\end{prop}

Proposition \ref{postopm} illustrates how a return policy can enhance price discrimination, allowing the seller's maximum revenue to reach the first-best surplus of $\mu_0 v$.\footnote{Another optimal mechanism would be a free-return policy with price $v$. Here we select the form given above, because a mechanism without cancellation fee is suboptimal when $\lambda_P$ is bounded away from $\infty$ or $v_l$ is bounded away from 0.  } In case 2 of Proposition \ref{postopm}, the seller sets the price to be  $q^{-1}(\mu_0)$ to make the buyer indifferent between learning before purchase and walking away. Hence, the buyer's option value is zero, $V^0(\mu_0;q^{-1}(\mu_0))=0$. Nevertheless, because he can obtain additional information value from post-purchase learning, he is willing to purchase the product first and learn afterward. However, the seller imposes a cancellation fee to extract this additional information value. Ultimately, the buyer receives zero participation value under the optimal mechanism. Note that the optimality of the second item in Proposition \ref{postopm} implies that the seller would benefit from providing more precise post-purchase information to the buyer, since this would enable her to extract more surplus through the cancellation fee.

\section{No news is good news}\label{bad}
In this section, we consider the opposite learning framework, introduced by \citet{badnews}, in which bad news arrives at rate $\rho$ if the buyer's true valuation is 0 (and no news arrives otherwise). We call this the bad-news model; the  previous framework is called the good-news model. For simplicity, we assume the learning rate is the same before and after purchase, and let the learning cost $k$ remain the same. Similarly, we assume $4k < v\rho$, so that learning is valuable for at least some prior beliefs.

Under the bad-news model, the buyer's posterior belief goes up as long as no news arrives, which leads to different learning behavior from that in our main model. Specifically, the buyer continues to learn either until  bad news arrives, at which point his belief falls to 0 and he requests a return, or until his belief becomes high enough that he decides to keep the product and not request a return. That is, he returns the product only if he is certain that his actual valuation is 0.  Consequently, the return policy has no direct impact on his stopping belief, whereas the price plays a major role in determining it. 

We use the same notation for the bad-news model as for our main model, except that we add the subscript $N$. 

\begin{lem}\label{IR-N}
    Under the optimal mechanism $m$, the buyer's IR binds: $V_N(\mu_0;m)=V^0_N(\mu_0;t_b)$. 
\end{lem}
Since the buyer  requests a return only if his belief is 0, his utility from a return under an optimal mechanism must be 0; in particular, $x_r \times 0-t_r=0$, so $t_r=0$.\footnote{If the return transfer were positive, his return utility would be negative and he would prefer to do all his learning before purchase to obtain his option value.} This means the magnitude of $x_r$ does not affect the seller's revenue; therefore, for the sake of exposition, we let $x_r=0$.\footnote{This assumption can be made without loss of generality as regards the optimal mechanism. } That is, if the seller ever allows returns, she uses the free-return policy $(0,0)$. 
Hence, the buyer's continuation value under the optimal mechanism depends only on the price. With slight abuse of notation, we use $V_N(\mu;t_b)$ to denote the buyer's value function. 

As in the main model, the buyer's learning strategy is characterized by  two cutoff beliefs, $q_N(t_b)<Q_N(t_b)$. Here $q_N(t_b)$ is the quitting belief, at which the buyer is indifferent between returning the product and continuing to learn, while $Q_N(t_b)$ is the consumption belief, at which he is indifferent between consuming the product and continuing to learn:
\[
Q_N(t_b)=1-\frac{k}{\rho t_b}.
\]

\begin{figure}[t]
        \centering
        \includegraphics[width=14cm]{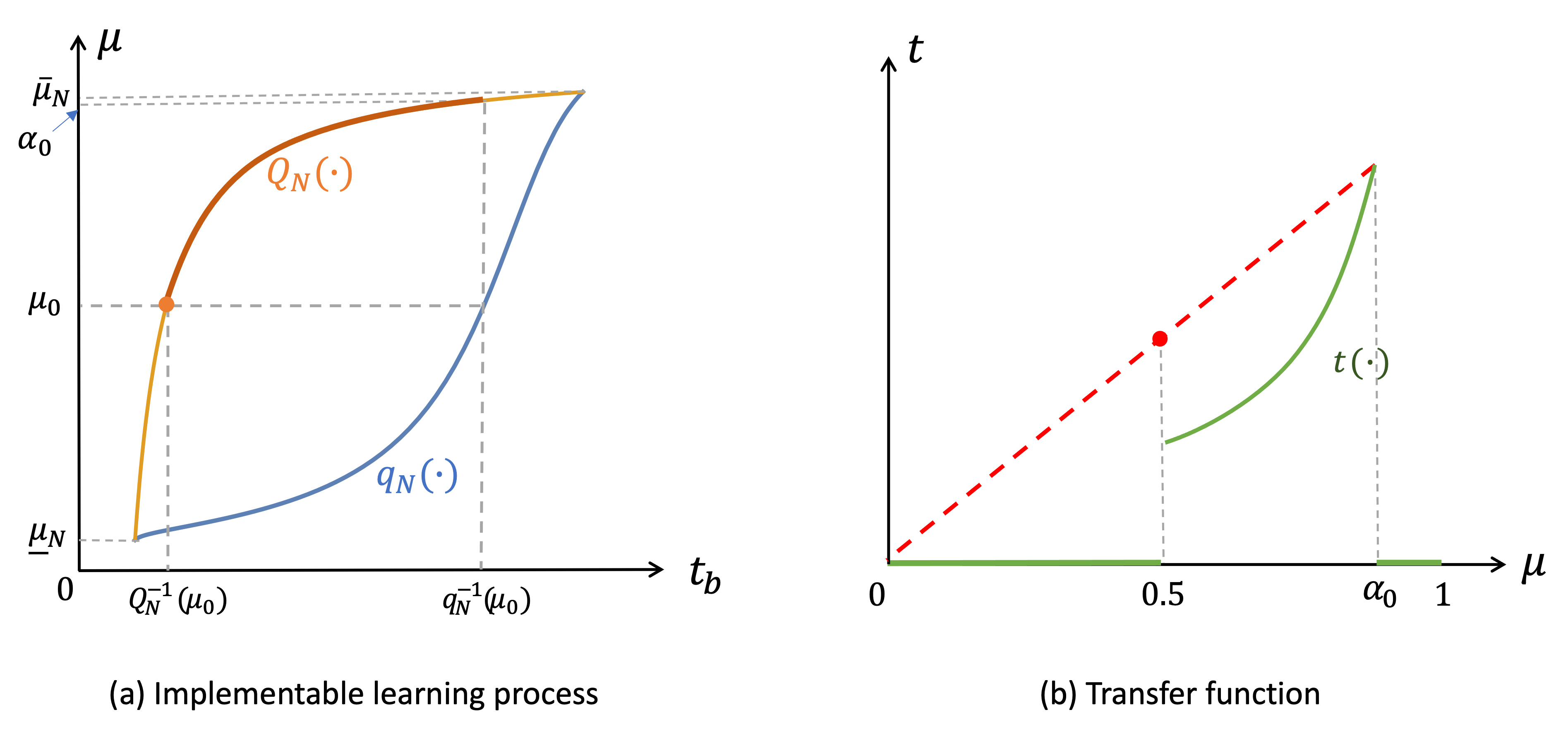}
        \caption{Bad-news model.}
        \label{transferbad}
        \end{figure}

Figure \ref{transferbad}(a) shows a graph of the quitting belief and consumption belief for the bad-news model. 
Since the buyer's posterior belief increases as long as no news arrives, every implementable learning process  corresponds to a distribution that has two atoms, one at 0 and the other at some $\alpha>\mu_0$. 
Specifically, let  $H^{\alpha}: [0,1]\rightarrow [0,1]$ be the following distribution, which has support $\{0,\alpha\}$: 
         \begin{equation*}
		  H^{\alpha}(\mu) = \begin{cases}
			1-\mu_0/\alpha &\text{if } \mu \in [0,\alpha), \\
			1 &\text{if } \mu \in [\alpha,1]. 
		\end{cases}
		\label{G}
	\end{equation*}
Unlike in the good-news model, here there is only one implementable stopping belief for a given fixed price $t_b$: the buyer always stops learning at the consumption belief $Q_N(t_b)$. Since he is willing to learn if the price lies in between $Q^{-1}_N(\mu_0)$ and $q^{-1}_N(\mu_0)$, the dark orange curve in Figure \ref{transferbad}(a) represents all implementable stopping beliefs. 
Moreover, $\alpha_0 := Q_N(q^{-1}_N(\mu_0))$ is the largest consumption belief that is implementable.  
Hence, the set of implementable distributions for the bad-news model is $\mathcal{H}^{\mu_0}=\{H^{\alpha}\}_{\alpha\in[\mu_0,\alpha_0]}$.

\begin{cor}
    A refund mechanism with price $t_b$ and a free-return policy is $H^{\alpha}$-implementable if 
     $t_b(\alpha)=Q^{-1}_N(\alpha)$, where  $\alpha\in (\mu_0,\alpha_0]$. 
  A refund mechanism with price $t_b= Q^{-1}_N(\mu_0)$ and a no-return policy  is $H^{\mu_0}$-implementable. 
\end{cor}

As in the good-news model, we construct the transfer function 
$t:[0,1] \rightarrow [0,v]$ piecewise: 
 \begin{equation*}
		  t(\mu) = \begin{cases}
			0 &\text{if } \mu \in [0,\mu_0), \\
			Q^{-1}_N(\mu) &\text{if } \mu \in [\mu_0,\alpha_0],  \\
			0 &\text{if } \mu \in (\alpha_0,1].
		\end{cases}
	\end{equation*}
For beliefs $\mu \in (\alpha_0,1]$ we let $ t(\mu) =0$, because those belief cannot be realized. Figure \ref{transferbad}(b) depicts a transfer function with $\mu_0=0.5$.  
The seller's problem is to maximize the expected transfer over the implementable distributions: 
\begin{equation*} 
\begin{aligned}
    \sup_{H^{\alpha}}\ \int_0^1 &t(\mu) \ dH^{\alpha} (\mu) \qquad  \\
  \textnormal{s.t. } \quad  \ & \ H^{\alpha}\in \mathcal{H}^{\mu_0} .
\end{aligned}
\end{equation*}
From Figure \ref{transferbad}(b), one can see that the second piece of the transfer function (the part supported on $[\mu_0,\alpha_0]$) is increasing and convex. Concavifying the transfer function over the implementable distributions $\mathcal{H}^{\mu_0}$ would imply a corner solution, in the sense that the optimal distribution would be either $H^{\mu_0}$ or $H^{\alpha_0}$ depending on the relative positions of the pieces of the transfer function. 

The next proposition summarizes the optimal refund mechanism for every prior belief; Figure \ref{badnews} provides a graphical representation.
\begin{prop}\label{negative}
    There is a non-empty closed interval $F_N \subset [\underline{\mu}_N,\overline\mu_N]$ such that the optimal mechanism takes the following form:
\vspace{0cm}
\begin{enumerate}\setlength{\itemsep}{0cm}
\item If $\mu_0 \notin[\underline\mu_N,\overline\mu_N]$, the mechanism offers a no-return policy with price $\mu_0 v$.
\item If $\mu_0 \in [\underline\mu_N,\overline\mu_N]$ and $\mu_0 \notin F_N$, the mechanism offers a no-return policy with the learning-deterrence price $Q^{-1}_N(\mu_0)$.
\item If $\mu_0 \in F_N$, the mechanism offers a free-return policy with price $Q^{-1}_N(\alpha_0)$.
\end{enumerate}
\end{prop} 
The form of the optimal refund mechanism for the bad-news model is similar to that of the good-news model. However, there is a crucial difference: when the seller optimally allows free returns, she price-discriminates to the extreme. As a result, the buyer adopts the longest learning process, $H^{\alpha_0}$, and gains zero participation utility, i.e., $V_N(\mu_0;Q^{-1}_N(\alpha_0))=0$.\footnote{From Figure \ref{transferbad}(a), the price that implements the stopping belief $\alpha_0$ is the same as the price $q^{-1}(\mu_0)$ that makes the buyer indifferent between learning and quitting at the prior belief. Hence the buyer's option value is 0. }  Consequently, the right endpoint of the interval $F_N$ coincides with $\overline\mu_N$. 
This is driven by the nature of the learning framework: in the bad-news model, as long as no news arrives, the buyer becomes increasingly optimistic and his continuation value rises, meaning the seller can raise the price until the buyer's initial continuation value is zero. By contrast, in the good-news model, as long as no news arrives, the buyer becomes increasingly pessimistic and his continuation value falls; when it hits zero, he requests a return. Therefore, to secure his participation, the seller has to provide him with a positive participation payoff. 
Thus, in the bad-news model, free returns further escalate price discrimination: the buyer can only obtain positive participation value if the optimal mechanism deters learning.

 \begin{figure}[t]
        \centering
        \includegraphics[width=11cm]{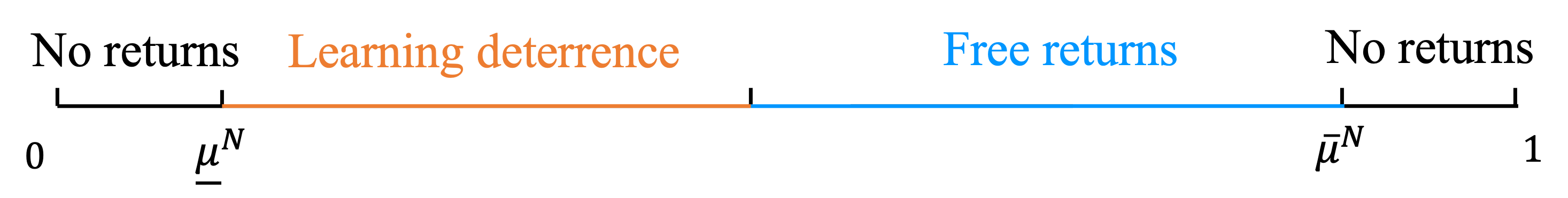}
        \caption{Optimal refund mechanism for bad-news model.}
        \label{badnews}
        \end{figure}

\section{Discussion}
\label{sec:discussion}

\subsection{Alternative interpretation}

In this section we discuss an alternative interpretation of the refund mechanism. Consider a start-up company offering an option contract, consisting of a baseline allocation and an option, to one acquirer. The baseline allocation specifies the price $t_r$ for a share $x_r$ of the company, and the acquirer has the option to purchase the remaining share ($1-x_r$) of the company at a strike price $(t_b-t_r)$. The timeline is as follows. At the beginning, the start-up offers this contract to the acquirer. The acquirer can then evaluate the contract and choose to accept or reject it. If he decides to accept it, then based on information that arrives afterward, he decides whether to purchase the remaining share. Notice that if the baseline allocation specifies a price for the entire company, i.e., $x_r=1$ and $t_b-t_r=0$, then the contract does not involve any real option.

With this interpretation, in the good-news model, the buyer purchases the remaining share only if he obtains conclusive good news before he stops learning (e.g., if some ongoing innovation of the start-up proves successful). Anticipating this, the start-up designs the baseline allocation so as to control the acquirer's learning decision, which hinges on his comparing the expected payoff from the baseline allocation and the continuation value from learning (the latter also depends on the strike price). By Theorem \ref{deterministic}, the optimal baseline allocation either allows the buyer to freely opt out or specifies a price for the entire company. Moreover, in the former case, the start-up designs a sufficiently high strike price to induce a relatively long learning process, while in the latter case the price is low enough to induce the buyer to purchase the entire company immediately. The optimal contract under the bad-news model has a similar structure.

\subsection{Remarks on the refund mechanism}
\textbf{Binary menu}. 
The refund mechanism we propose is essentially a binary menu with two options, $\{(1,t_b),(x_r,t_r)\}$; we implicitly assume that the first option allocates the product to the buyer with probability one. This is an assumption without loss of generality. To see why, suppose the seller offers a mechanism $\{(x_b,t_b),(x_r,t_r)\}$ with $0\le x_r\le x_b<1$. The buyer's option value depends only on the high-type surplus, which is $x_b v-t_b$ in this case. Therefore, if $x_b<1$, the seller can increase this allocation rate and simultaneously increase the price so as to keep the  high-type surplus constant. This would generate a profitable deviation, since the buyer's option value, his quitting belief, and the optimal stopping belief would all remain the same (because they depend only on the high-type surplus), but the price would be higher.\footnote{Suppose the seller encourages learning. If we keep the high-type surplus constant, then reducing $x_b$ can only reduce the trial belief. This will shrink the set of implementable distributions and thus have only a downward effect on revenue. If the seller deters learning, then one can easily show that the learning-deterrence price is increasing in $x_b$. } 

\textbf{Multi-option menu}. In principle, the refund mechanism can consist of arbitrarily many options, $\{(1,t_b),(x^1_r,t^1_r),\ldots,(x^K_r,t^K_r)\}$, where $K\ge 1$ is an integer. One can easily obtain a binding IR constraint in this case.
Moreover, the obedience constraint (OB) implies that the buyer is willing to stop learning at $\beta^k$ and request the return option $(x^k_r,t^k_r)$.  Both IR and OB imply that the buyer is indifferent to stopping learning at all $\beta^k$ with $1\le k\le K$. From an ex-ante point of view, the seller can therefore implement a distribution $G$ with finite support $\{\beta^1,...,\beta^K,1\}$.\footnote{Note that if the support has more than two elements, then the distribution is not uniquely determined through its support under a binary state space. }  However, a standard result in \cite{KamenicaGentzkow2011} indicates that with a binary state space, a distribution with binary support is sufficient to characterize the solution of $\max_{F} \int_0^1 y(x) d F(x)$, where $F$ is the set of feasible distributions on the posterior beliefs. Hence, it is without loss to consider a refund mechanism with only one return option.

\textbf{Return window}. As mentioned in Section \ref{sec:model}, the seller cannot benefit from imposing a return window; that is, her maximum revenue remains the same whether or not her policy is allowed to include a return window. To elaborate, suppose the optimal refund mechanism is characterized by the triplet $\left\{t_b,\left(x_r, t_r\right)\right\}$ and the return window $T$. We distinguish between two cases: either (a) the buyer stops learning (weakly) before reaching the return window and will not change his learning behavior if the return window is extended, or (b) he stops learning at the return window but would be willing to keep learning if the return window were extended. In case (a) it is trivially true that the seller cannot gain from imposing the return window. In case (b), the seller has a profitable deviation. To see this, suppose the buyer's learning is restricted through $T$, and let his posterior belief upon reaching the return window be $\mu_T$.\footnote{In the case $\lambda_P=\lambda_B$, the return window cannot effectively limit the buyer's learning process. For example, the buyer can calculate the stopping time $\tau^*$ that would be optimal in the absence of the return window, and learn for a duration of $\tau^*-T$ before purchase. Hence $\mu_T=\mu(\tau^*)$. In the case $\lambda_P>\lambda_B$, $\mu_T$ is more complicated. To explicitly characterize $\mu_T$, we need to solve for the buyer's optimal purchase time, denoted by $\tau_P$. This also determines his belief at purchase, denoted by $\mu_P=\mu(\tau_P)$. The buyer's stopping belief at the return window $\mu_T$ is determined by the law of motion for post-purchase learning with a starting belief $\mu_P$ from $\tau_P$ to $\tau_P+T$. Nevertheless, there is a well-defined stopping belief at the return window. } His value function for post-purchase learning, $V_P(\,\cdot\;m,T)$, is the solution to \eqref{odeb} with a boundary point $\left(\mu_T, x_r \mu_T v-t_r\right)$. Because he would prefer to continue leaning if $T$ were extended, his value function on the whole belief space has a kink at $\mu_T$. To the left of $\mu_T$, the value function coincides with the return utility $x_r \mu v-t_r$. Note that the kink is an upward kink, i.e., the return utility crosses the value function for post-purchase learning from above. Hence the seller can set the same return window $T$ but increase both $x_r$ and $t_r$ to make the return utility tangent to the value function. The return window will then have no bite. In other words, we can first solve for the optimal refund mechanism without a return window and then set the return window equal to the buyer's stopping time under the optimal mechanism.

This is consistent with the observation  that during the holiday season, retailers often provide extended return windows so that customers can take their time trying out products. For example, an Amazon announcement for the 2023 holiday season states, ``Although the return window for most orders will be extended, returns eligibility for all orders remains the same.'' 
In other words, unless the seller does not allow returns at all, she will not use the return window to restrict the buyer's learning.

\newpage
\section*{Appendix A}\label{appendixa}

\begin{proof}[{\bf Proof of Lemma \ref{set}}] 
Suppose that stopping learning at $\beta<q(t_b)$ is an optimal learning strategy given a mechanism $\{t_b,(x_r,t_r)\}$. Then such a stopping belief is determined by the smooth pasting and value matching condition. Smooth pasting implies  $V_1(\beta;t_b)\ge 0$ since  $x_r\ge 0$ and thereby the return payoff is weakly increasing in belief. From (\ref{odeb}), we can conclude $V(\beta;t_b)<0$. Contradiction. 
\end{proof}

\begin{proof}[{\bf Proof of Proposition \ref{benchprop}}] The  optimal learning strategy in case 2 is a corollary of Proposition 3.1 in \cite{KRC}. To prove the existences of $\underline{t}_b$ and $\overline{t}_b$, we want to show $q(t_b)\le Q(t_b)$ if $t_b\in[\underline{t}_b,\overline{t}_b]$ with equality holds at the two boundaries. At the quitting belief, $q(\underline{t}_b)$,  the buyer must prefer to walk away rather than accept the price,
\[
q(t_b) v-t_b\le 0. 
\]
This inequality leads to the existence of $\underline{t}_b$ and $\overline{t}_b$ given the assumption $4 k< \lambda v$.  

Let $\tilde{\mu}(\mu):=q(Q^{-1}(\mu))$ be the quitting belief when the price is the inverse of the trial belief, $t_b=Q^{-1}(\mu)$. First, notice that $\tilde{\mu}(\cdot)$ is increasing in $\mu$ because both $q(\cdot)$ and $Q(\cdot)$ are increasing. Moreover,  $\tilde{\mu}(\underline{\mu})=\underline{\mu}$, $\tilde{\mu}(\bar\mu)=\bar\mu$. We want to show $\tilde{\mu}(\mu)<\mu$ if $\mu\in(\underline{\mu},\bar\mu)$. 
With the construction of the trial belief, 
\[V(\mu;Q^{-1}(\mu))=\mu v-Q^{-1}(\mu). \]
Take implicit differentiation with respect to to $\mu$, we obtain
\[
 \frac{d Q^{-1}(\mu)}{d \mu}=-\frac{k (k-\lambda (v-Q^{-1}(\mu)))}{\lambda^2(1-\mu)^2\mu (v-Q^{-1}(\mu))}=-\frac{k (\tilde{\mu}-1)}{\lambda(1-\mu)^2\mu }>0.
\]
Furthermore, 
\[
\frac{d \tilde{\mu}}{d \mu}=-\frac{k^2(k-\lambda (v-Q^{-1}(\mu)))}{\lambda^3(1-\mu)^2 \mu (v-Q^{-1}(\mu))^3}=\frac{\tilde{\mu}^2 (1-\tilde{\mu})}{(1-\mu)^2 \mu}.
\]
Thus, $\tilde{\mu}(\mu)$ is a differential equation with initial point $(\underline{\mu},\underline{\mu})$,
and its solution is
\begin{equation}\label{mutilde}
     -\frac{1}{\tilde{\mu}} + \log\left[\frac{\tilde{\mu}}{1-\tilde{\mu}}\right]=\frac{1}{1-\mu} +\log\left[\frac{\mu}{1-\mu}\right]-\frac{\lambda v}{k}.
\end{equation}
The left-hand side is a function of $\tilde{\mu}$ and the right-hand side is a function of $\mu$. Both are increasing and they coincide at $\underline\mu$ and $\bar\mu$. Moreover, the slope of the left-hand side is larger (smaller) than the right-hand side if both arguments are smaller (larger) than 0.5. Given that $\tilde{\mu}(\mu)$ is determined when the left-hand side equals the right-hand side,  thus, $\tilde{\mu}(\mu)<\mu$ if $\mu\in(\underline\mu,\bar\mu)$. Since $Q^{-1}(\mu)$ is monotone in $\mu$, then $q(t_b)\le Q(t_b)$ if $t_b\in[\underline{t}_b,\overline{t}_b]$.
\end{proof}

\begin{proof}[{\bf Proof of Lemma \ref{newlem1}}]
Suppose $V_P(\mu_0;m)\ge V^0(\mu_0;t_b)$ holds with strict inequality. That is,  under the mechanism $m$, 
if the buyer optimally stops learning at  belief $\beta$ and requests a return, he gets a payoff $V_P(\beta;m)>V^0(\beta;t_b)$. 
With the standard smooth-pasting and value-matching conditions,  
\[
 t_r(\beta) =\beta V'_P(\beta; m)/v-V_P(\beta;m).
\]
Note that $V'_P(\beta;m)$ and $V_P(\beta;m)$ are constrained by the differential equation 
\[
(1-\beta)\beta \lambda V_P'(\beta;m)+ \beta \lambda V_P(\beta;m)=\beta \lambda (v-t_b) - k
\]
derived from the Bellman equation for post-purchase learning. As long as $V_P(\beta;m)> V^0(\beta;t_b)$, 
the seller can reduce $V_P(\beta;m)$ and increase $V'_P(\beta;m)$  with respect to the differential equation, to induce the same stopping belief $\beta$. This leads to a higher $t_r(\beta)$ and thereby implies a profitable deviation. 
\end{proof}

\begin{proof}[{\bf Proof of Proposition \ref{implementable}}]
Given Lemma \ref{newlem1}, then standard smooth pasting and value matching conditions give the formulas for $x_r$ and  $t_r$ when $\beta\in [q(t_b),\mu_0)$. We want to show the second equality in equation (\ref{value}). Take integration by part and omit the notation in the brackets,
\[
\beta V_1^0 -V^0 = \beta V^0_1 -\int_0^{\beta} V^0_1 d\mu =\beta V^0_1-V^0_1 \mu \biggm|^\beta_{q(t_b)}+\int^{\beta}_{q(t_b)} \mu d V^0_1=\int^{\beta}_{q(t_b)} \mu V^0_{11} d \mu,
\]
where $V^0_{11}$ is the second-order derivative with respect to $\beta$. The second equality comes from integration by part and $V^0_1=0$ for $\beta\le q(t_b)$. Now we show $\mu V^0_{11}=k/|\mu'(\tau)|$ to close the proof. Rearrange the (\ref{odeb}),
\[
k=\mu \lambda(v-t_b-V^0)-\lambda\mu(1-\mu)V^0_1.
\]
Take total differentiation with respect to $\mu$ and multiply both sides by $\mu$,
\[
\mu V^0_{11}\lambda \mu (1-\mu)=\mu \lambda  (v-t_b-V^0)-\lambda \mu (1-\mu) V^0_1=k,
\]
where $\lambda \mu (1-\mu)=|\mu'(\tau)|=1/|\tau'(\mu)|$. Thus, $
\mu V^0_{11}=k |\tau'(\mu)|=MC(\mu)$, and 
\[
t_r(\beta,t_b)= \int^{\beta}_{q(t_b)} \mu V^0_{11} d \mu =\int^{\beta}_{q(t_b)} MC(\mu) d \mu.
\]
\end{proof}

\begin{proof}[{\bf Proof of Proposition  \ref{step1}}]
We want to show that $t_b^*=Q^{-1}(0.5)$ is the solution to $\beta^o(t^*_b)=Q(t^*_b)$ and $ t_b^{**}=v/2$ is the solution to $\beta^o(t^{**}_b)=q(t^{**}_b)$. 

First, the construction of $t^*_b$ implies that, 
\[
t^*_b=\beta^o(t^*_b)  v- V^0. 
\]
Since $t_r(\beta^o(t^*_b),t^*_b)=\beta^o(t^*_b) V^0_1-V^0$, we obtain, 
\[
t^*_b-t_r=\beta^o(t^*_b) \left(\frac{V^0-(\beta^o(t^*_b) v-t^*_b)}{1-\beta^o(t^*_b)}+\frac{k}{|\mu'(\tau)|} \right)=\beta^o(t^*_b) \frac{k}{|\mu'(\tau)|},
\]
where $k/|\mu'(\tau)|$ is evaluated at $\beta^o(t^*_b)$.
The first equality is from  a slight rearrangement of the (\ref{odeb}). The second equality comes from the first item in the bracket equals 0. Then substitute the expression into equation (\ref{foc}) and simplify it, we obtain, 
\[
\frac{1}{\beta^o(t^*_b)}=\frac{1}{1-\beta^o(t^*_b)} \Longrightarrow \beta^o(t^*_b)=0.5 .
\]
Hence, $t^{*}_b=Q^{-1}(0.5)$. 

Next, the construction of $t^{**}_b$ implies, $t_r(\beta^o(t^{**}_b);t^{**}_b)=0$ and $t^{**}_b=v-\frac{k}{\lambda \beta^o(t^{**}_b)}$. Substitute them back to equation (\ref{foc}),
\[
\frac{k/(\lambda \beta^o(t^{**}_b)(1-\beta^o(t^{**}_b)))}{v-k/(\lambda \beta^o(t^{**}_b))}=\frac{1}{1-\beta^o(t^{**}_b)} \Longrightarrow \beta^o(t^{**}_b)=\frac{2k}{\lambda v}.
\]
Hence, $t^{**}_b=v/2$.  

If $\mu_0=0.5$, to deter learning, the seller sets a non-refundable price $Q^{-1}(0.5)$; to encourage learning, the price must stay in $(Q^{-1}(0.5),q^{-1}(0.5)]$. Since $t_b^*=Q^{-1}(0.5)$ and $t_b^{**}<q^{-1}(0.5)$, we obtain the optimal information process in proposition \ref{step1}. 
\end{proof}

\begin{cor}\label{fullcharacterization} 
The optimal information process $G^{\beta^*}$ is determined as following: 
\vspace{-0.2cm}
\begin{itemize}\setlength\itemsep{0cm}
    \item[1.] If $\mu_0\le q(t_b^{**})$, then $\beta^*(t_b)\rightarrow \mu_0$ regardless of the price $t_b\in [Q^{-1}(\mu_0),q^{-1}(\mu_0)]$. 
    \item[2.] If $\mu_0\in (q(t_b^{**}),0.5)$, then,  
    \vspace{-0.05cm}
    \begin{enumerate}\setlength\itemsep{0cm}
 \item[(a).] If $t_b\in [Q^{-1}(\mu_0), t^*_b]$, then the $\beta^*(t_b)\rightarrow\mu_0$;
    \item[(b).] If $t_b\in (t^*_b,t^{**}_b)$, then the $\beta^*(t_b)=\beta^o(t_b)$;
    \item[(c).] If $t_b\in [t^{**}_b,q^{-1}(\mu_0)]$, then $\beta^*(t_b)=q(t_b)$. 
\end{enumerate}
\item[3.] If $\mu_0\in (0.5,Q(t^{**}_b))$, then,  
 \vspace{-0.05cm}
\begin{enumerate}\setlength\itemsep{0cm}
 \item[(a).] If $t_b=Q^{-1}(\mu_0)$, then the $\beta^*(t_b)=\mu_0$;
    \item[(b).] If $t_b\in (Q^{-1}(\mu_0),t^{**}_b)$, then the $\beta^*(t_b)=\beta^o(t_b)$;
    \item[(c).] If $t_b\in [t^{**}_b,q^{-1}(\mu_0)]$, then $\beta^*(t_b)=q(t_b)$. 
    \end{enumerate}
    \item[4.] If $\mu_0\ge Q(t^{**}_b)$, then $\beta^*(t_b)=q(t_b)$ regardless of the price $t_b\in [Q^{-1}(\mu_0),q^{-1}(\mu_0)]$. 
\end{itemize}
\end{cor}

\begin{proof}
    The optimal information process in 1 and 2 is immediate after anticipating the constraints $t_b\in [Q^{-1}(\mu_0),q^{-1}(\mu_0)]$ and $\beta\in [q(t_b),\mu_0]$. For 3 and 4, when $\mu_0>0.5$, then the transfer function is convex when $\mu\in(0.5,\mu_0)$. Hence, $\int_0^1 t(\mu,t_b) d G^{\beta}(\mu)$ can have a local maximizer when $\mu\le 0.5$ and another local minimizer when $\mu>0.5$. If it is true, then the global maximizer $\beta^*(t_b)$ can be either the local maximizer $\beta^o(t_b)$ or $\beta^*(t_b)\rightarrow \mu_0$. However, if the latter case is true, then using stochastic return to implement $G^{\mu_0^{-}}$ is strictly dominated by using learning deterrence to implement $G^{\mu_0}$, where $\mu_0^{-}$ represents a belief  converges to $\mu_0$ from the left.  Therefore, in any case, setting a price $t_b>Q^{-1}(\mu_0)$ and implementing $G^{\mu_0^{-}}$ can not be an optimal refund mechanism. However, to close the proof for this corollary, we want to show the seller's expected revenue $\int_0^1 t(\mu,t_b) d G^{\beta}(\mu)$ is actually quasi-concave in $\beta$, and therefore we rule out the latter case. Note that this part of the proof is not necessary for the characterization of the optimal refund mechanism in the next section. 

    To show $\int_0^1 t(\mu,t_b) d G^{\beta}(\mu)$ is quasi-concave in $\beta$ when $\beta\in[q(t_b),\mu_0)$, we show a relaxed version that the objective function is quasi-concave in $\beta$ when $\beta\in[q(t_b),Q(t_b)]$. To abuse notation, let $\beta^o(t_b)$ be the solution to equation (\ref{foc}). It can be a correspondence. Note that the proof of Proposition \ref{step1} suggests that the local optimizer $\beta^o(t_b)$ of the objective function (either maximizer or minimizer) only coincides with $Q(t_b)$ when $t_b=Q^{-1}(0.5)$. Moreover, $(\beta^o)'>0$ if $\beta^o>0.5$, $(\beta^o)'<0$ if $\beta^o<0.5$, and $(\beta^o)'\rightarrow \infty$ at $\beta^o=0.5$. Hence, if there is a  local minimizer, it is always larger than $Q(t_b)$. 
\end{proof}

\begin{proof}[{\bf Proof of Lemma \ref{elasticity}}]
    With simple rearrangement, 
\[
\mathcal{E}_r(\beta^o(t_b))=\frac{t_b}{1-\beta^o(t_b)}(\beta^o)'(t_b), \textnormal{ and } \mathcal{E}_r(q(t_b))=\frac{t_b}{1-q(t_b)}q'(t_b). 
\]
Recall that $(\beta^o)'=\frac{(\beta^o)^2 (1-\beta^o)}{(2 \beta^o-1)(1-q(t_b))k/\lambda}$ from footnote \ref{footnote}, and $q'(t_b)=q(t_b)/(v-t_b)$. Hence, 
\[
\frac{-\mathcal{E}_r(\beta^o(t_b))}{\mathcal{E}_r(q(t_b))}=\frac{(\beta^o(t_b))^2}{q(t_b)^2 (1-2 \beta^o  )}>1,
\]
because $q(t_b)\le \beta^o(t_b)\le 0.5 $. 
\end{proof}

\begin{proof}[{\bf Proof of Theorem \ref{two}}]
Theorem \ref{deterministic} implies that the optimal refund mechanism is either learning deterrence or free return. We first characterize the optimal free return mechanism. The revenue from a free return mechanism only depends on the price, 
\begin{equation}\label{freereturn}
    \sup_{t_b}\quad  \int_0^1 t(\mu,t_b) d G^{q(t_b)}(\mu)=\frac{\mu_0-q(t_b)}{1-q(t_b)} t_b, 
\end{equation}
\[
\textnormal{s.t }\quad  Q^{-1}(\mu_0)<t_b\le q^{-1}(\mu_0).
\]
First, ignore the constraint on price and solve the unconstrained optimal price, denoted as $t_b^F$,
\[
t_b^F=v-\frac{k}{\lambda}-\left(\frac{k}{\lambda}\left(v-\frac{k}{\lambda}\right)\frac{1-\mu_0}{\mu_0}\right)^{1/2}.
\]
Denote $\Pi^F$ as the expected revenue of a free return mechanism with price $t^F_b$. Note that both $t^F_b$ and $\Pi^F$ are functions of $\mu_0$. 

We now verify that the price constraint is satisfied when $\Pi^F (\mu_0)\ge Q^{-1}(\mu_0)$, where $Q^{-1}(\mu_0)$ is the revenue/price from learning deterrence. Suppose $\Pi^F (\mu_0)\ge Q^{-1}(\mu_0)$, then $t^F_b (\mu_0)> Q^{-1}(\mu_0)$ because the probability of a successful sale is less than one with free return. Moreover, $t^F_b(\mu_0)\le q^{-1}(\mu_0)=v-k/\lambda \mu_0$ can be verified by simple algebra given that $q(v)<\mu_0$. 

Next, we want to show that $F$ is either an empty set or a closed interval. Note that the revenue from learning deterrence is strictly higher than the optimal free return when the prior belief is either below $\underline\mu$ or above $\bar\mu$. Hence, we only need to show that $\Pi^F(\mu_0)$ intersects with $Q^{-1}(\mu_0)$ at most twice when $\mu_0\in (\underline\mu,\bar\mu)$.  That is, $\Pi^F(\mu_0)=Q^{-1}(\mu_0)$ has at most two roots in $(\underline\mu,\bar\mu)$. Substitute this in equation (\ref{mutilde}) and let $c\equiv k/\lambda$, 
\begin{equation}\label{slope}
    \frac{1}{1-\mu_0} +\log\left[\frac{\mu_0}{1-\mu_0}\right]-\frac{ v}{c}+\frac{1}{q(\Pi^F(\mu_0))} -\log\left[\frac{q(\Pi^F(\mu_0))}{1-q(\Pi^F(\mu_0))}\right]=0
\end{equation}
The first-order derivative of the left-hand side of the above equation with respect to $\mu_0$ is,
\begin{equation*}
\frac{1}{1-\mu_0}\left(\frac{1}{1-\mu_0}+\frac{1}{\mu_0}\right)+\left(\frac{1}{\sqrt{\mu_0}}-\frac{r}{\sqrt{1-\mu_0}}\right)\frac{(\sqrt{\mu_0}+ r \sqrt{1-\mu_0})^3}{(\sqrt{\mu_0}+r \sqrt{1-\mu_0})^2-1},
\end{equation*}
where $r=\sqrt{v/c-1}>\sqrt{3}$. 
Let $x\equiv \sqrt{\frac{\mu_0}{1-\mu_0}}\in (0,\infty)$, which is a monotone transformation of $\mu_0$. Let the left-hand side of the above equation equal to 0 and rearrange it,  
$$m(x):=\frac{x (x+r)^3 (1-x r)}{(1+x^2)^3 (-1+2x r+r^2)}=-1,$$
where $m(x)$ is a rational function. The degree of the numerator is smaller than that of the denominator, thus it has a horizontal asymptote $m=0$. Note that the denominator is positive, hence it does not have a vertical asymptote. Meanwhile $\lim_{x\to 0} m(x)=0$, $\lim_{x\to \infty} m(x)=0$, $m(x=1)<0$, and $m(x)=0$ has a unique root $x=1/r<1$. Therefore, the graph of $m(x)$ is the following.
\begin{figure}[H]
        \centering
        \includegraphics[width=7cm]{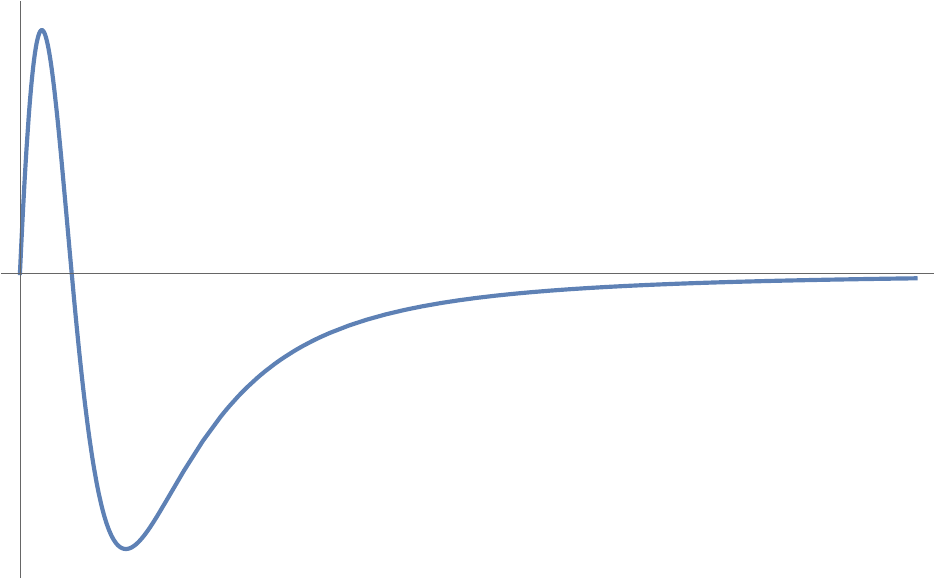}
        \end{figure}
Then, $m(x)=-1$ has at most two roots. That is, if there are two roots, then the left-hand side of the equation (\ref{slope}) is first increasing, then decreasing, and then increasing again. Given that the left-hand side of equation (\ref{slope})  is strictly positive at $\underline{\mu}$ and $\overline{\mu}$, then equation (\ref{slope}) has at most two roots, and thereby $F$ is either empty or a non-empty interval. Moreover, if $m(x)=-1$ has no or one  root, then $F$ is empty. 

The existence of $c^*$ comes from  that $\Pi^F(\mu_0)$ decreases in $c=\frac{k}{\lambda}$ and $Q^{-1}(\mu_0)$ increases in $c$. The former one can be verified from the Envelop theorem of the objective function (\ref{freereturn}). The latter one can be verified by the Implicit function theorem of $\mu_0 v-Q^{-1}(\mu_0)=V^o(\mu_0;Q^{-1}(\mu_0))$ with respect to $c$. 
\end{proof}

\begin{proof}[{\bf Proof of Proposition \ref{postopm}}]
First, recall that $V^o(\mu_0;t_b)$ is the solution to the (\ref{odeb}) with boundary point $(q(t_b),0)$, 
\begin{equation}\label{V0}
    V^o(\mu_0;t_b)=\mu_0 (v-t_b)-\frac{k}{\lambda}-(1-\mu_0)\frac{k}{\lambda}\left[\log\left(\frac{\mu_0}{1-\mu_0}\right)-\log\left(\frac{q(t_b)}{1-q(t_b)}\right)\right].
\end{equation}
Similarly, $V_P(\mu_0;m^c)$ is the solution to the (\ref{odeb}) with boundary point $(q_P(t_b,t_c),-t_c)$, replacing $\lambda$ by $\lambda_P$,
\[V_P(\mu_0;m^c)=\mu_0 (v-t_b)-\frac{k}{\lambda_P}-(1-\mu_0)\frac{k}{\lambda_P}\left[\log\left(\frac{\mu_0}{1-\mu_0}\right)-\log\left(\frac{q_P}{1-q_P}\right)\right]-(1-\mu_0) t_c.\]
Slight rearranging the above equation,  we have,
\begin{equation*}
    \begin{aligned}
        \frac{k}{\lambda_P} \log\left(\frac{q_P}{1-q_P}\right)-t_c&=\frac{k}{\lambda_P} \log\left(\frac{\mu_0}{1-\mu_0}\right)+\frac{k}{\lambda_P}\frac{1}{1-\mu_0}-\frac{\mu_0}{1-\mu_0}(v-t_b)+\frac{V_P(\mu_0;m^c)}{1-\mu_0}\\
        &=\frac{k}{\lambda_P} \log\left(\frac{\mu_0}{1-\mu_0}\right)+\frac{k}{\lambda_P}\frac{1}{1-\mu_0}-\frac{\mu_0}{1-\mu_0}(v-t_b)+\frac{V^0(\mu_0;t_b)}{1-\mu_0},
    \end{aligned}
\end{equation*}
where the second equality comes from Corollary \ref{post}. 

Since $t_r(\beta,t_b)=\int_{q_P}^{\beta} \frac{k}{\lambda_P \mu(1-\mu)} d\mu= \frac{k}{\lambda_P}\left[\log\left(\frac{\beta}{1-\beta}\right)-\log\left(\frac{q_P}{1-q_P}\right)\right]$,
\[
t_r(\beta,t_b)+t_c(t_b)=\frac{k}{\lambda_P}\log\left(\frac{\beta}{1-\beta}\right)-\frac{k}{\lambda_P}\log\left(\frac{\mu_0}{1-\mu_0}\right)-\frac{k}{\lambda_P}\frac{1}{1-\mu_0}+\frac{\mu_0}{1-\mu_0} (v-t_b)-\frac{V^o(\mu_0;t_b)}{1-\mu_0}. 
\]
If $\lambda_P\rightarrow \infty$, we obtain,
\[
t_r+t_c=\frac{\mu_0}{1-\mu_0} (v-t_b)-\frac{V^o(\mu_0;t_b)}{1-\mu_0}. 
\]
We consider two cases: (a) $t_b\in [q^{-1}(\mu_0), q^{-1}(\mu_0;\lambda_P)]$ where $V^o(\mu_0;t_b)=0$; and (b) $t_b\in [Q^{-1}(\mu_0),q^{-1}(\mu_0))$, where $V^o(\mu_0;t_b)$ is specified as in equation (\ref{V0}).

In case (a), the expected revenue  $\mu_0 t_b+(1-\mu_0)(t_r+t_c)=\mu_o v$ is a constant. In case (b), the expected revenue can be simplified to,
\[
\mu_0 t_b+(1-\mu_0)(t_r+t_c)=\mu_0 t_b +\frac{k}{\lambda}+(1-\mu_0)\frac{k}{\lambda}\left(\log\left(\frac{\mu_0}{1-\mu_0}\right)-\log\left(\frac{q(t_b)}{1-q(t_b)}\right)\right),
\]
whose first-order derivative with respect to $t_b$ simplifies to, 
\[
\mu_0-(1-\mu_0)\frac{q(t_b)}{1-q(t_b)}>0. 
\]
Thus, the revenue is increasing in $t_b$, rendering $q^{-1}(\mu_0)$ the optimal price. Moreover, $t_r+t_c=\frac{k}{\lambda(1-\mu_0)}$. The optimal revenue equals $\mu_0 v$ and buyer's surplus $V^o(\mu_0;q^{-1}(\mu_0))=0$. 

\end{proof}

\begin{proof}[{\bf Proof of Lemma \ref{IR-N}}]
Negative transfer at return is not optimal. 
\end{proof}

\begin{proof}[{\bf Proof of Proposition \ref{negative}}]
Suppose $4k<v\rho$, $\underline{\mu}^N$ and $\overline{\mu}^N$ are determined as the two roots such that the buyer at the consumption belief is indifferent between consuming and walking away, $\mu v-Q^{-1}_N(\mu)=0. $ Thus, if the prior belief is either smaller than $\underline{\mu}^N$  or bigger than $\overline{\mu}^N$, then learning is suboptimal regardless of the mechanism and thereby the optimal refund mechanism follows form 1. 

When $\mu_0\in [\underline{\mu}^N,\overline{\mu}^N]$, 
 the optimal mechanism either implements $H^{\mu_0}$ or $H^{\alpha_0}$, we compare the revenue from both cases to determine the optimal mechanism. 
Suppose implementing $H^{\mu_0}$ via learning deterrence dominates implementing $H^{\alpha}$ via free return. Equivalently, 
\[
Q^{-1}_N(\mu_0)\ge \frac{\mu_0}{\alpha}Q^{-1}_N(\alpha) \Longleftrightarrow \alpha(1-\alpha)\ge \mu_0(1-\mu_0).
\]
Then if $\mu_0<(>)1/2$, then implementing $H^{\mu_0}$ dominates (is dominated by) $H^{\alpha_0}$ as long as $\alpha_0\in (\mu_0,1-\mu_0)$ ($\alpha_0\in (1-\mu_0,\mu_0)$).  

Recall that $\alpha_0$ is the solution to 
\[
V_N(\mu_0;Q^{-1}_N(\alpha_0))=-\frac{k}{\rho} -\mu_0 \frac{k}{\rho (1-\alpha_0)} +\mu_0 v-\mu_0 \frac{k}{\rho}\left(\log\left(\frac{\alpha_0}{1-\alpha_0}\right)-\log\left(\frac{\mu_0}{1-\mu_0}\right)\right)=0.
\]
By the implicit function theorem, 
\[
\alpha'_0(\mu_0)=\frac{(1-\alpha_0)^2\alpha_0}{\mu_0^2(1-\mu_0)}>0. 
\]
Moreover, since $\alpha_0(\underline\mu)=\underline\mu$ and  $\alpha_0(\bar\mu)=\bar\mu$, we have $\alpha'_0(\underline\mu)>1$ and  $\alpha'_0(\bar\mu)<1$. One can verify that $\alpha_0(\mu_0)$ crosses $1-\mu_0$ once from below if $\mu_0\in(\underline\mu,\bar\mu)$. Thus, we have form 2 and 3. 
\end{proof}

\newpage 
\section*{Appendix B}\label{appendixb}
In this section, we provide complementary materials to derive the cost of learning in order to clarify our marginal cost of information and derive the continuation value for learning.  
Recall that in our model, the buyer's true valuation is either high or low, and good news  arrives according to the Poisson rate $\lambda$ if the true valuation is high. For the sake of exposition, we call the buyer's true valuation the state, $\theta \in\{0,1\}$, where $0$ represents low valuation and $1$ represents high valuation. If the buyer continues to learn from time $\tau$ to  $\tau+d\tau$ without receiving any good news, then his belief evolves according to,
\[
\mu(\tau+d\tau)=\frac{\mu(\tau)(1-\lambda d\tau)}{\mu(\tau)(1-\lambda d\tau)+1-\mu(\tau)},
\]
which leads to the law of motion $\mu'(\tau)=-\mu(\tau)(1-\mu(\tau))\lambda$, a well-behaved differential equation with initial point $\mu(\tau=0)=\mu_0$. The solution to this differential equation is,
\[
\mu(\tau)=\left(1+\frac{1-\mu_0}{\mu_0}e^{\lambda \tau}\right)^{-1}. 
\]
Hence, conditional on no news arriving, the buyer's posterior belief is a deterministic function of time. 

Consider a learning strategy such that the buyer stops learning either when his posterior belief reaches $\beta\le \mu_0$ or $1$. Under such a strategy, we can calculate how long it takes for the buyer to reach the belief $\beta$ conditional on no news arriving. We denote such time duration as $T_{0}$ and it satisfies $\mu(T_0)=\beta$. Hence, 
\begin{equation}\label{T0}
    T_0=\frac{1}{\lambda}\left[\log\left(\frac{\mu_0}{1-\mu_0}\right)-\log\left(\frac{\beta}{1-\beta}\right)\right].
\end{equation}
Note that $T_0$ is also the time cost conditional on $\theta=0$ if the buyer follows such a learning strategy. Hence, the information cost under such learning strategy conditional $\theta=0$ is $kT_0$, which is a function of both the posterior belief $\beta$ and the prior belief $\mu_0$. 

Note that,
\[
k T_0=\int_{\beta}^{\mu_0} MC(\mu) d\mu.
\]  
Therefore, when the seller applies stochastic return to implement a stopping belief $\beta$ higher than the quitting belief, the gain in the return transfer is the saved learning cost under state $0$.

To determine the buyer's continuation value from learning, we want to calculate the expected time that a Poisson jump arrives before $T_0$ (when the buyer stops learning), denoted as $\mathbb{E}(T|T\le T_0,1)$. To abuse notation, we use $T$ to represent the time that good news arrives. $T$ is a random variable.  
Note that 
\[
\frac{1}{\lambda}=e^{-\lambda T_0}\mathbb{E}(T|T>T_0,1) +(1-e^{-\lambda T_0}) \mathbb{E}(T|T\le T_0,1),
\]
where the left-hand side is the expected time that good news arrives in state $1$. In the right-hand side, $e^{-\lambda T_0}$ is the expected probability that no news arrives before time $T_0$, and $\mathbb{E}(T|T>T_0,1)$ is the expected time that good news arrives after $T_0$. Thus, $\mathbb{E}(T|T>T_0,1)=T_0+\frac{1}{\lambda}$. Moreover, $(1-e^{-\lambda T_0}) $ is the probability that good news arrives before $T_0$. Thus, we can solve,
\[
\mathbb{E}(T|T\le T_0,1)=\frac{1}{\lambda}-\frac{e^{-\lambda T_0}}{1-e^{-\lambda T_0}}T_0.
\]
Let  $\mathbb{E}[T_1]$ be the expected stopping time in state $1$ conditional on the learning strategy. It is the sum of expected stopping time in both  events:   no news arrives until $T_0$ and good news arrives before $T_0$,
\begin{equation}\label{T1}
   \mathbb{E}[T_1]=e^{-\lambda T_0}T_0+(1-e^{-\lambda T_0}) \mathbb{E}(T_1|T_1\le T_0,1)=\frac{1}{\lambda}\frac{\mu_0-\beta}{\mu_0(1-\beta)}.  
\end{equation}
Using $T_0$ and $\mathbb{E}[T_1]$, we can calculate the buyer's continuation value from an ex-ante point of view, conditional on a particular learning strategy parameterized by $\beta$. We use $U^{\beta}(\mu_0)$ to denote it.
\[
 U^{\beta}(\mu_0) =\mu_0(1-e^{-\lambda T_0}) (v-t_b)-\mu_0 k \mathbb{E}[T_1]-(1-\mu_0) k T_0. 
\]
Essentially, $U^{\beta}(\mu_0)$ is equal to the expected gain from a successful sale net the expected cost from learning. 
$V(\mu_0;t_b)$ in the main text can be obtained if we replace  $\beta$  by the quitting belief $q(t_b)$. Doing so,  
\[
V(\mu_0;t_b)=U^{q(t_b)}(\mu_0)=\mu_0 (v-t_b)-\frac{k}{\lambda}-(1-\mu_0) \frac{k}{\lambda}\left(\log\left(\frac{\mu_0}{1-\mu_0}\right)-\log\left(\frac{q(t_b)}{1-q(t_b)}\right)\right).
\]

\newpage
\bibliographystyle{aer}
\bibliography{orm}

\end{document}